\documentclass[preprint]{elsarticle}

\usepackage{pifont}
\usepackage{etex}
\usepackage{booktabs,tabularx}
\usepackage{natbib}
\usepackage{geometry}
\usepackage{graphicx}
\usepackage{hyperref}
\usepackage{amsmath}
\usepackage{amsthm}
\usepackage{amssymb}
\usepackage{mathrsfs}
\usepackage[all]{xypic}
\usepackage{nicefrac}
\usepackage{hyperref}
\usepackage[shortlabels]{enumitem}
\usepackage{mathtools}
\usepackage{wasysym}
\usepackage{stmaryrd}
\usepackage{tikz}
\usepackage[utf8]{inputenc}
\usepackage{aliascnt}
\usepackage{MnSymbol}
\usepackage{microtype}

\usetikzlibrary{arrows,shapes,snakes,automata,backgrounds,petri,calc,positioning}

\renewcommand\appendixname{}

\SelectTips{cm}{}

\DeclareMathOperator{\Hom}{Hom}

\newcommand{\trans}[1]{\xrightarrow{#1}}
\newcommand{\enames}{\mathcal{E}}
\newcommand{\causes}{\mathscr{K}}
\newcommand{\cat}[1]{\mathbf{#1}}
\newcommand{\catPO}{\cat{P}}
\newcommand{\catO}{\mathbb{O}}
\newcommand{\catPOm}{{\cat{P}_{m}}}
\newcommand{\catSet}{\cat{Set}}
\newcommand{\catGraph}{\cat{Graph}}
\newcommand{\pshO}{\catSet^\catO}
\newcommand{\catC}{\cat{C}}
\newcommand{\catD}{\cat{D}}
\newcommand{\psh}[1]{\catSet^{#1}}
\newcommand{\pshC}{\psh{\catC}}
\newcommand{\catCoalg}[1]{{#1}\textbf{-Coalg}}
\newcommand{\catONSet}{\catO\textbf{-Set}}
\newcommand{\catSym}[1]{\textbf{Sym}(#1)}
\newcommand{\pshOPP}{\pshO_{\diamond}}
\newcommand{\deltaFun}{\boldsymbol{\delta}}
\newcommand{\pshMark}{\mathscr{M}}
\newcommand{\lpsh}{\mathcal{L}}
\newcommand{\parts}{\mathcal{P}}
\newcommand{\finParts}{\parts_f}
\newcommand{\Parts}{\mathscr{P}}
\newcommand{\FinParts}{\Parts_f}
\newcommand{\homFun}[1]{\Hom_{#1}}
\newcommand{\oilts}{$\catO$-ILTS}
\newcommand{\oiltsC}{\oilts$_{\mathtt{AC}}$}
\newcommand{\ctrans}[3]{\trans{#1 \caus {#2}_{#3}}}
\newcommand{\actrans}[2]{\xRightarrow{#1 \caus #2}}
\newcommand{\ictrans}[3]{\,|\!\!\!\xRightarrow[#3]{#1 \caus #2}}
\newcommand{\icar}{\ar@{|=>}}
\newcommand{\icstrans}[3]{\,\Vert\!\!\!\xRightarrow[#3]{#1 \caus #2}}
\newcommand{\icsar}{\ar@{||=>}}
\newdir{X>}{%
:(0,0){\hspace{-1pt}\RHD}}
\newcommand{\Trans}[1]{\xymatrix@1{\ar@{=X>}[r]^{#1} &}}

\newcommand{\acTrans}[2]{\Trans{#1 \caus #2}}
\newcommand{\bstrans}[2]{\trans{#2}_{\raisebox{1.5ex}{\scriptsize${#1}$}}}
\newcommand{\caus}{\vdash}
\newcommand{\subto}{\hookrightarrow}
\newcommand{\injto}{\rightarrowtail}
\newcommand{\dclos}[2]{{{#1}\!\!\downarrow_{#2}}}
\newcommand{\minp}[1]{\llbracket #1 \rrbracket}
\newcommand{\ext}[1]{{#1}^{\!+}}
\newcommand{\isorep}[1]{[#1]_{\cong}}
\newcommand{\funext}[3]{#1[\nicefrac{#2}{#3}]}
\newcommand{\el}{\mathord{\int}}

\newcommand{\pre}[1]{{}^\bullet{#1}}
\newcommand{\post}[1]{{#1}^\bullet}
\newcommand{\fires}[1]{\mathbin{[{#1}\rangle}}

\newcommand{\bsbisim}{\sim_{bs}}
\newcommand{\cb}{\ensuremath{\mathtt{C}}}
\newcommand{\acb}{\ensuremath{\mathtt{AC}}}
\newcommand{\icb}{\ensuremath{\mathtt{IC}}}
\newcommand{\icsb}{\ensuremath{\mathtt{ICS}}}

\newcommand{\cbisim}{\sim_{\cb}}
\newcommand{\acbisim}{\sim_{\acb}}
\newcommand{\icbisim}{\sim_{\icb}}
\newcommand{\icsbisim}{\sim_{\icsb}}

\newcommand{\clts}{CG$_{\mathtt{C}}$}
\newcommand{\rclts}{CG$_{\mathtt{C}}^r$}
\newcommand{\aclts}{CG$_{\mathtt{AC}}$}
\newcommand{\iclts}{CG$_{\mathtt{IC}}$}
\newcommand{\icslts}{CG$_{\mathtt{ICS}}$}
\newcommand{\bsc}{B\ensuremath{_{\mathtt{C}}}}

\newcommand{\dpos}[2]{[#1]_{#2}}
\newcommand{\dposT}[3]{[#1]_{#2}^{#3}}

\newcommand{\imm}[1]{{#1}_{\mathtt{I}}}

\newcommand{\pushoutcorner}[1][dr]{\save*!/#1+1.2pc/#1:(1,-1)@^{|-}\restore}
\newcommand{\pullbackcorner}[1][dr]{\save*!/#1-1.2pc/#1:(-1,1)@^{|-}\restore}

\newcommand{\qqquad}{\hspace{4em}}
\newcommand{\minpfun}[1]{\langlebar #1 \ranglebar}
\newcommand{\psym}[3]{{#1} \RHD_{#2} {#3}}
\newcommand{\norm}[2]{\Vert #1 \Vert_{#2}}
\newcommand{\symrep}[1]{\Lbag \! #1 \! \Rbag}

\newcommand{\lss}{\preccurlyeq}

\newcommand{\mynewtheorem}[2]{
  \newaliascnt{#1}{dummy}
  \newtheorem{#1}[#1]{#2}
  \aliascntresetthe{#1}
  \expandafter\def\csname #1autorefname\endcsname{#2}
}

\newcommand{\notationname}{Notation}

\theoremstyle{plain}
  \mynewtheorem{theorem}{Theorem}
  \mynewtheorem{proposition}{Proposition}
  \mynewtheorem{corollary}{Corollary}
  \mynewtheorem{property}{Property}
  \mynewtheorem{lemma}{Lemma}
  \mynewtheorem{conjecture}{Conjecture}
\theoremstyle{definition}
  \mynewtheorem{definition}{Definition}
  \mynewtheorem{example}{Example}
\theoremstyle{remark}
  
  \mynewtheorem{remark}{Remark}
  \mynewtheorem{convention}{Convention}
  \mynewtheorem{invariant}{Invariant}

\begin{document}

\title{A coalgebraic semantics for causality in Petri nets}

\author[di]{Roberto Bruni}
\ead{bruni@di.unipi.it}
\ead[url]{http://www.di.unipi.it/~bruni}

\author[di]{Ugo Montanari}
\ead{ugo@di.unipi.it}
\ead[url]{http://www.di.unipi.it/~ugo}

\author[ru]{Matteo Sammartino\corref{cor}}
\ead{m.sammartino@cs.ru.nl}
\ead[url]{http://www.cs.ru.nl/M.Sammartino/}

\tnotetext[proj]{Research supported by the EU Integrated Project 257414 ASCENS, by the Italian MIUR Project CINA (PRIN 2010LHT4KM) and by the NWO Project 612.001.113 Practical Coinduction.}
\cortext[cor]{Corresponding author. Tel.: (+31) 0243652642}
\address[di]{University of Pisa, Computer Science Department, Largo Bruno Pontecorvo 3, 56127 Pisa, Italy}
\address[ru]{Radboud University, Institute for Computing and Information Sciences, Faculty of Science, Heyendaalseweg 135,
6525AJ Nijmegen, The Netherlands}

\begin{abstract}
In this paper we revisit some pioneering efforts to equip Petri nets with compact operational models for expressing causality. The models we propose have a bisimilarity relation and a minimal representative for each equivalence class, and they can be fully explained as coalgebras on a presheaf category on an index category of partial orders. First, we provide a set-theoretic model in the form of a a \emph{causal case graph}, that is a labeled transition system where states and transitions represent markings and firings of the net, respectively, and are equipped with causal information. Most importantly, each state has a poset representing causal dependencies among past events. Our first result shows the correspondence with behavior structure semantics as proposed by Trakhtenbrot and Rabinovich. Causal case graphs may be infinitely-branching and have infinitely many states, but we show how they can be refined to get an equivalent finitely-branching model. In it, states only keep the most recent causes for each token, are up to isomorphism, and are equipped with a symmetry, i.e., a group of poset isomorphisms. Symmetries are essential for the existence of a minimal, often finite-state, model. This first part requires no knowledge of category theory. The next step is constructing a coalgebraic model. We exploit the fact that events can be represented as names, and event generation as name generation. Thus we can apply the Fiore-Turi framework, where the semantics of nominal calculi are modeled as coalgebras over presheaves. We model causal relations as a suitable category of posets with action labels, and generation of new events with causal dependencies as an endofunctor on this category. Presheaves indexed by labeled posets represent the functorial association between states and their causal information. Then we define a well-behaved category of coalgebras. Our coalgebraic model is still infinite-state, but we exploit the equivalence between coalgebras over a class of presheaves and \emph{History Dependent automata} to derive a compact representation, which is equivalent to our set-theoretical compact model. Remarkably, state reduction is automatically performed along the equivalence.
\end{abstract}

\begin{keyword}
Petri nets, causal case graph, behavior structures, presheaves, coalgebras, HD-automata.
\end{keyword}

\maketitle

\section{Introduction}

Petri Nets are a well-known graphical and formal notation for representing concurrent computations. An interesting aspect of Petri Nets is that they allow for the representation of causal dependencies among actions. This kind of information can be useful for debugging distributed systems or for tracing expected or unwanted causal dependencies, and it is usually not provided by interleaving models.

In order to carry out verification on Petri nets, it is convenient to have an \emph{operational} model, that is a model representing single steps of computation and their observable actions. In Petri nets, steps are typically firings and actions are action labels of transitions. One important class of operational models for Petri Nets are \emph{behavior structures} \cite{TR88}. They are automata where each state is equipped with a partial order over events: events represent different occurrences of actions and the poset describes causal dependencies among such occurrences. Behavior structures come with a notion of behavioral equivalence, which later has been called \emph{history preserving bisimilarity} \cite{FroschleH99}.

Other causal models, such as \emph{event structures} \cite{NielsenPW81}, do not come with a built-in operational notion of bisimilarity. Such a notion is essential to compute minimal models, where all states with the same behavior are identified. Open maps \cite{JoyalNW96} can be used to derive \emph{hereditary history preserving bisimulations} (HHPBs), but the existence of minimal representatives is not guaranteed by that theory. Indeed, the general agreement is that HHPB is more suited to capture concurrency, whereas the non-hereditary version deals better with causality. The latter equivalence is coarser, but still causality is informative enough to characterize key security properties, such as non-interference \cite{BaldanCa14}. Moreover, the non-hereditary equivalence has better decidability properties than the hereditary one~\cite{FroschleH99}.

The main issue with causal operational models is that they often have infinitely many states, so model checking is unfeasible. This is indeed the case of behavior structures, where posets of states are enlarged at each transition, because a new event for the corresponding action is generated. Even if we minimize w.r.t.\ bisimilarity, there is no way of throwing away ``useless'' events or decreasing the size of posets.

In this paper we present an approach to obtain compact, and in many cases finite, operational models for causality in Petri nets. They will be presented in two ``flavors'': a set-theoretic and a categorical one, based on coalgebras  \cite{Rutten00,Adamek05}. 
In addition to the theoretical and practical interest of reconducting our problem to unifying and well studied models such as coalgebras, we emphasize that our coalgebraic model is simpler than the set theoretical one. In fact, even if deriving a naive set-theoretic model from a Petri net is not difficult, the technical development required to obtain a compact model is quite involved and requires some ingenuity. Instead, in a categorical setting, this machinery will become remarkably simpler and natural. Actually, in a precise sense, the construction of the compact model will be automatic, thus providing a mathematical justification of the otherwise ad hoc set-theoretic constructions.

\subsection{Set-theoretic models}

After some preliminaries on Petri nets and the presentation of a running example in \autoref{sec:basic},
in \autoref{sec:lts} we model the behavior of a labeled Petri net as a \emph{causal case graph} (CG). Recall that a case graph is a labeled transition graph where states are markings and transitions are steps, representing many firings happening simultaneously. In causal case graphs, transitions are single firings, and causal data are used to encode information about concurrency. More precisely (see \autoref{def:clts}, where CGs are called ``concrete'' as opposed to ``abstract'' CGs, introduced later):
\begin{itemize}
    \item states are of the form $O \rhd c$, where: $O$ is a poset describing causal dependencies among a finite collection of events; $c$ is a marking where each token is decorated with its \emph{causes}, i.e. the set of events that led to its creation (included in $O$);

\item the transition relation is written $\ctrans{K}{e}{a}$, where: $K$ is the set of most recent causes of tokens that enabled the firing; $e$ is a \emph{fresh} event, different from all those occurring in the source state; and $a$ is the action label of the fired transition.

\end{itemize}

\begin{table}[p]
	\centering
	\begin{tabularx}{\textwidth}{ X | X }
		\multicolumn{1}{c|}{States} & \multicolumn{1}{c}{Transition relation}
		\\
		\midrule
		\multicolumn{2}{c}{{\bf Causal case graph (CG)}} \\
		\midrule
		\[
			O \rhd c
		\]
		\begin{itemize}[leftmargin=*]
			\item $O$ is a finite poset describing causal dependencies among events
			\item $c$ is a marking including causes for each token
		\end{itemize}
		&
		\[
			\ctrans{K}{e}{a}
		\]
		\begin{itemize}[leftmargin=*]
			\item $K$ is the set of most recent causes of tokens consumed by the transition
			\item $e$ is a fresh event
			\item $a$ is the fired transition's action label
		\end{itemize}		
		\\
		\midrule
		\multicolumn{2}{c}{{\bf Abstract CG (\aclts)}} \\
		\midrule
		\[
			O \rhd c
		\]
		\begin{itemize}[leftmargin=*]
			\item $O$ is a \emph{canonical representative} of isomorphic posets
			\item $c$ contains canonical events
		\end{itemize}
		&
		\[
			\actrans{K}{a}
		\]
		\begin{itemize}[leftmargin=*]
			\item $K$ as in CG
			\item $a$ is the action label for the canonical fresh event
		\end{itemize}
		\\
		\toprule
		\multicolumn{2}{c}{{\bf Immediate causes CG (\iclts)}} \\
		\midrule
		\[
			O \RHD c
		\]
		\begin{itemize}[leftmargin=*]
			\item $O$ and $c$ contain only the most recent causes w.r.t.\ each token (\emph{immediate causes})
			\item each state is a canonical representative of isomorphic states
		\end{itemize}
		&
		\[
			\ictrans{K}{a}{h}
		\]
		\begin{itemize}[leftmargin=*]
			\item $K$ and $a$ as in \aclts
			\item $h$ is a map telling how events in the target state correspond to those of the source state
		\end{itemize}
		\\
		\toprule
		\multicolumn{2}{c}{{\bf Immediate causes CG with symmetries (\icslts)}} \\
		\midrule
		\[
			\psym{O}{\Phi}{c}
		\]
		\begin{itemize}[leftmargin=*]
			\item $O$ and $c$  as in \iclts
			\item $\Phi$ is a symmetry on $O$
		\end{itemize}
		&
		\[
			\icstrans{K}{a}{h}
		\]
		\begin{itemize}[leftmargin=*]
			\item $K$,$a$ and $h$ as in \iclts
			\item transitions are canonical representatives of ``symmetric'' ones
		\end{itemize}
	\end{tabularx}
	\caption{Set-theoretic models.}
	\label{tab:tr-ref}
\end{table}

We define a notion of bisimilarity for CGs where causal information plays a key role: only states with the same causal dependencies among past events, namely the same poset, are compared. This fact is crucial for the equivalence with history preserving bisimilarity described in \autoref{sec:bs}.

Another important aspect is that transitions draw fresh events from an infinite set of event names. For each firing, we have \emph{infinitely many} transitions in the CG, one for each possible fresh event. In this way we implement \emph{event generation} in the same way name generation is represented, e.g., in nominal calculi. This fact will be crucial for our categorical models.

We, then, derive three consecutive refinements of the CG, described in \autoref{tab:tr-ref}, each improving the CG on one aspect:

\begin{description}
	\item[\aclts~(\autoref{def:aclts}):] the transition relation becomes \emph{finitely branching}, because we don't distinguish between posets with the same structure. In fact, it is enough to generate one canonical event, instead of all possible ones, for each firing. Consequently, states contain canonical representatives of events and only the action label of the new event is recorded in the transition.

	\item[\iclts~(\autoref{def:iclts}):] removing all but immediate causes, and identifying isomorphic states, may significantly reduce the state space, and even make it finite.

	\item[\icslts~(\autoref{def:icslts}):] we equip each state with a set of isomorphisms acting as the identity on the state. These isomorphisms must form a \emph{symmetry}, i.e., a group of automorphisms, on the state's poset. Transitions are reduced accordingly: we select one representative for each collection of ``symmetric'' transitions. Two transitions are symmetric whenever they can be obtained from each other via isomorphisms belonging to the symmetries of source and target states. Symmetries allow for the computation of minimal models, because CGs that are not isomorphic, but bisimilar under a given isomorphism, have a unique minimal realization, where that isomorphism becomes part of the symmetry of a state.
\end{description}
These steps do not change the overall semantics (Theorems \ref{thm:acbisim-cbisim} and \ref{thm:corr-ic-ac}). 

Finally, in \autoref{thm:bs-c-bisim} we establish a connection between CGs and behavior structures.

\subsection{Categorical models}

In the second part of the paper (Sections \ref{sec:back}-\ref{sec:hd}) we assume the reader has some familiarity with category theory. Some preliminaries about presheaves and coalgebras are recalled in \autoref{sec:back}.

Coalgebras are convenient models of dynamic systems. Their theory is rich and well-developed, and many kinds of systems have been characterized in this setting. Coalgebras are also of practical interest: minimization procedures such as \emph{partition refinement} \cite{KanellakisS90}  can be defined in coalgebraic terms (see, e.g., \cite{AdamekBHKMS12}). This further motivates the coalgebraic framework: algorithms implemented at this level of abstraction can be instantiated to many classes of systems.

Our coalgebraic causal model of Petri nets, presented in \autoref{sec:coalg}, is based on the fact that we represent events as names and event generation as name generation, in the style of nominal calculi. This allows us to construct a coalgebra where states are equipped with nominal structures, namely causal relations between events, and event generation is explicit, along the lines of \cite{FioreT01}. The key idea is to define coalgebras over \emph{presheaves}, that are functors from a certain \emph{index category} $\catC$ to $\catSet$, the category of sets and functions. Presheaves formalize the association between a collection of names, seen as an objects of $\catC$, and a set of processes within $\catSet$, indexed by names of the collection. Fresh name generation can be formalized as an endofunctor on $\catC$, that is lifted to presheaves and used in the definition of coalgebras.

We take as index category for presheaves a suitable category of \emph{labeled posets} up to isomorphism, representing causal relations between events decorated with actions. This category provides us with the needed structure to model operations over causal relations. In fact, we use colimits to implement a \emph{well-behaved functorial} model of event generation, which augments a given poset with fresh events and relations to their causes. Our definition ensures that its lifting to presheaves, when used to define coalgebras, yields a category of coalgebras with a final object and a final semantics in agreement with coalgebraic bisimilarity. This is essential for a correct notion of minimal model. Then, we define a presheaf of causal markings, yielding, for each poset, the set of causal markings whose causes are ``compatible'' with that poset. We construct a \emph{causal coalgebra} by translating the abstract CG. The important result is that coalgebraic and ordinary bisimilarity are equivalent (\autoref{thm:bisim-equiv}).

The infinite state issue still exists in the causal coalgebra, because the poset of a causal marking keeps growing along transitions. However, if the presheaf of states is ``well-behaved'', according to \cite{CianciaKM10}, it is always possible to recover the \emph{support} of a causal marking, that is the minimal poset including all and only events that appear in the marking. This is the key condition for the equivalence between presheaf-based coalgebras and \emph{History Dependent (HD) automata} \cite{Pistore99}.

HD-automata are coalgebras with states in \emph{named-sets} \cite{CianciaM10}, that are sets whose elements are equipped with symmetry groups over finite collections of names. They have two main features:
\begin{itemize}
	\item a single state can represent the whole orbit of its symmetry, namely all the states reachable via poset isomorphisms;
	\item the names of each state are \emph{local}, related to those of other states via suitable mappings. 
\end{itemize}
Both features are important for applying finite state methods, such as minimization and model-checking, to nominal calculi. 
In particular, the latter point captures \emph{deallocation}: maps between states can discard unused names and ``compact'' remaining ones, much like \emph{garbage collectors} do for memory locations. A minimization procedure for HD-automata for the (finite-control) $\pi$-calculus has been shown and implemented in \cite{FerrariMT05}.

Interestingly, we are able to define the presheaf of causal markings in a way that computing the support corresponds to discarding all but the immediate causes.
Therefore, in \autoref{sec:hd} we show that the aforementioned equivalence amounts to deriving the immediate causes CG. Actually, it also equips states with symmetries, achieving the last refinement step. We emphasize that such equivalence is completely standard in the theory of nominal calculi. In our case, it is extended to labeled posets and allows the automatic derivation of an HD-automaton over a named set of minimal causal markings.

\section{Basic definitions and running example}
\label{sec:basic}

Given a set of labels $L$, we call \emph{$L$-labeled poset} (or just labeled poset, when $L$ is clear from the context) on a set $S$ a triple $O = (X_O,\lss_O,l_O)$, where $X_O \subseteq S$, $\lss_O$ is a reflexive, transitive and antisymmetric relation on $X_O$ and $l_O \colon X_O \to L$ is a labeling function. A morphism of labeled posets $O \to O'$ is a function $\sigma \colon X_{O} \to X_{O'}$ that preserves order and labeling, namely $x \lss_O y$ implies $\sigma(x) \lss_{O'} \sigma(y)$ and $l_O = l_{O'} \circ \sigma$. We say that $\sigma$ \emph{reflects order} whenever $\sigma(x) \lss_{O'} \sigma(y)$ implies $x \lss_O y$; $\sigma$ is an \emph{order-embedding} whenever it both preserves and reflects order. Notice that isomorphisms reflect order, because their inverses preserve order, and it can be easily checked that order-embeddings are always injective. To simplify notation, we sometimes regard $O$ as a poset on $S \times L$, we write $|O|$ for the underlying set of pairs and $x_l \in X_O \times L$ for the pair $(x,l) \in |O|$. A set $K \subseteq |O|$ is \emph{down-closed} w.r.t.\ $O$ whenever $y \in K$ and $x \lss_O y$ implies $x \in K$. We say that a poset $O$ is a \emph{prefix} of $O'$ if $O$ is a subposet of $O'$ and $|O|$ is down-closed w.r.t.\ $O'$.

In this paper we consider the following kind of Petri nets, which we call just \emph{nets}.
\begin{definition}[Net]
A net is a tuple $(S,T,F,l)$ where:
\begin{itemize}
    \item $S$ is a set of \emph{places} and $T$ is a set of \emph{transitions}, with $S \cap T= \emptyset$;
    \item $F \subseteq (S \times T) \cup (T \times S)$ is the \emph{flow relation};
    \item $l \colon T \to Act$ is a \emph{labeling function}, where $Act$ is a fixed set of action labels.
\end{itemize}
\end{definition}
If $x \in S \cup T$ then $\pre{x} = \{ y \mid (y,x) \in F \}$ and $\post{x} = \{ y \mid (x,y) \in F \}$ are called the \emph{pre-set} and \emph{post-set} of $x$, respectively; for all $t \in T$, we assume $\pre{t}, \post{t} \neq \emptyset$.
 A \emph{marking} $m$ is a multiset over $S$. 
A transition $t \in T$ is \emph{enabled} at marking $m$ if $s \in m$, for all $s \in \pre{t}$, in which case it can \emph{fire}, written $m \fires{t} m'$, i.e., a new marking $m' = (m \setminus \pre{t}) \cup \post{t}$ is produced. We say that a net is \emph{marked} whenever it has an \emph{initial marking} $m_0$. We denote by $\fires{m_0}$ the set of markings reachable from $m_0$ by a (finite) sequence of firings.

We require that elements of initial markings have multiplicity one. This implies that $m_0$ is actually a set, in agreement with the fact that pre-sets and post-set in nets are sets, meaning that they can only consume one token at a time from a given place. In typical P/T nets transitions may consume many tokens from the same place, but this difference is inessential for the development of our theory. 

\paragraph{\bf Running example}
\label{ex:running}
As a running example, we will use the marked net defined as follows: $S = \{s_1,s_2\}$, $T= \{t_1,t_2,t_3\}$, $F$ includes $(s_i,t_i),(s_i,t_3)$ (for $i=1,2$) and symmetric pairs, and $l(t_1) = l(t_2) = a$, $l(t_3) = b$. The initial marking is $m_0 = \{s_1,s_2\}$. This net is depicted below: circles denote places, squares denote transitions, edges describe the flow relation, and filled circles indicate the position of tokens in $m_0$. Notice that $\fires{m_0} = \{m_0\}$.

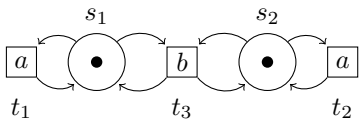
\begin{center}
\begin{tikzpicture}[auto]    
    \node[place,label=$s_1$,tokens=1] (s1) {};
            
    \node[place,label=$s_2$,tokens=1,right=1.5cm of s1] (s2) {};          
    \coordinate (mid) at ($(s1)!0.5!(s2)$);
    
    \node[transition] (t3) at (mid) {$b$}
        edge [pre,bend right=45] node {} (s1)
        edge [pre,bend left=45] node {} (s2)        
        edge [post,bend right=45] node {} (s2)
        edge [post,bend left=45] node {} (s1);
    
    \node[below=1ex of t3] (l3) {$t_3$} ;
    
    \node[transition] (t1) [left of=s1] {$a$}
        edge [pre,bend left=45] node {} (s1)
        edge [post,bend right=45] node {} (s1);
    
    \node[below=1ex of t1] (l1) {$t_1$} ;
        
    \node[transition] (t2) [right of=s2] {$a$}
        edge [pre,bend right=45] node {} (s2)
        edge [post,bend left=45] node {} (s2);
        
    \node[below=1ex of t2] (l2) {$t_2$} ;        
\end{tikzpicture}
\end{center}

\section{Causal semantics for Petri nets}
\label{sec:lts}

In this section we introduce our \emph{causal labeled semantics} for nets. It will be in the form of a \emph{causal case graph} (CG in short), that is a labeled transition graph whose states are markings with causal information and transitions represent firings. We start from a naive CG, derived from a given net in the simplest way, and then we give three subsequent refinements that will lead to a compact and, in some cases, finite-state CG. 
Throughout this section we fix a net $N = (S,T,F,l)$ and we assume that an infinite set $\enames$ of event names (or just events) is available.

The key idea is to equip markings with information about the occurrences of actions that led to the creation of each token. An occurrence of a transition labeled by $a \in Act$ is represented as an $Act$-labeled event $e_a$. Formally, a \emph{causal marking} $c$ is a set of the form
\[
    \{ K_1 \caus s_1,\dots,K_n \caus s_n  \}
\]
where $K_i \subseteq \finParts(\enames \times Act)$ is the set of \emph{causes} of $s_i \in S$, for $i=1,\dots,n$. More specifically, if $e_a \in K_i$ then the sequence of firings that generated the token includes a transition with action label $a$. We write $\causes(c)$ for $K_1 \cup \dots \cup K_n$ and $|c|$ for the underlying marking $\{s_1,\dots,s_n\}$ of $c$. Given a marking $m$ and $K \subseteq \finParts(\enames \times Act)$, $K \caus m$ is the causal marking obtained by assigning causes $K$ to each $s \in m$.

Transitions of our CGs will generate new events and their causal dependencies. In order to keep track of these data, we equip causal markings with $Act$-labeled posets, describing the causal relations between events which are occurrences of past actions. 
\begin{definition}[P-marking]
A P-marking is a pair $O \rhd c$, where $c$ is a causal marking and $O$ is a finite $Act$-labeled poset on $\enames$ such that: if $K \caus s \in c$ then $K$ is down-closed w.r.t.\ $O$.
\end{definition}
Down-closure requires each set of causes to contain the whole ``history'' of its events, as described by $O$. Nevertheless, $O$ may contain events that are unrelated to or caused by those of $\causes(c)$, but that are not among them.

Posets will have different purposes in the different classes of CGs we are going to introduce: they will be used to record either all the events happened so far or the ``most recent'' ones. The shape of P-markings will not change, but there will be additional requirements on their components.

We introduce a useful operation on P-markings. Their posets can be enlarged by adding events from which existing events causally depend on, but a closure operator must be applied, in order to retain down-closure of sets of causes. 
\begin{definition}[Closure operator]
Given $K \subseteq |O|$ and $O'$ such that $O$ is a subposet of $O'$, the \emph{closure} of $K$ w.r.t.\ $O'$ is given by
\[
	\dclos{K}{O'} = \bigcup_{x \in K} \{ y \in |O'| \mid y \lss_{O'} x \}
\]
Its extension to causal markings is $\dclos{(K \caus s)}{O'} = \dclos{K}{O'} \caus s$ and acts element-wise on sets.
\end{definition}
Given a P-marking $O \rhd c$ and $O' \supseteq O$, it can be easily verified that $O' \rhd \dclos{c}{O'}$ is a proper P-marking.

\subsection{Concrete CG}

The first step is deriving a CG from the net. Its states are P-markings $O \rhd c$ such that $O$ contains the whole history of past events and transition labels are of the form $K \caus e_a$, meaning that an $a$-labeled transition $t$ is fired: $e_a$ is an event fresh w.r.t.\ all the previous ones (i.e., those in $O$) and $K$ is the set of most recent causes associated to tokens that enabled $t$. We call this CG \emph{concrete} because posets with the same structure but different event names are distinguished. 
\begin{definition}[Concrete CG]
\label{def:clts}
The \emph{concrete CG} (\clts) is the smallest CG generated by the following rule
    \[
        \frac{
            t \in T
            \quad
            |c| = \pre{t}
            \quad
            a = l(t)
            \quad
            e \in \enames \setminus X_O
            \quad
            K = \max_O \causes(c)
        }
        {
            O \rhd c \cup c' \ctrans{K}{e}{a} \delta(O,K,e_a) \rhd (\causes(c) \cup \{e_a\} \caus \post{t} ) \cup c'
        }
    \]
where $\max_O K$, for $K \subseteq |O|$, is the set of maximal elements in $K$ according to $O$, and $\delta(O,K,x) =  ( O \cup (K \times \{ x \}) )^*$.

\end{definition}
Given a P-marking, the rule above checks whether it includes a causal marking $c$ such that its underlying marking is the pre-set of a transition $t$ ($|c| = \pre{t}$). If this is the case, $t$ is turned into a CG transition whose label $K \caus e_a$ is formed by the maximal causes $K$ of $c$ w.r.t.\ $O$ and by a labeled event $e_a$, where $e$ does not occur in the source poset ($e \notin \enames \setminus X_O$). The target state is obtained by replacing $c$ with the tokens produced by the firing, each equipped with the whole set of causes of $c$ plus the new event $e_a$. Since $e_a$ is causally dependent on the causes of $c$, the poset in the target state is updated with new pairs representing such dependencies by taking $\delta(O,K,e_a)$.

Note that event generation is similar to name generation in nominal calculi.\footnote{The relationship between $\pi$-calculus and causality has been investigated in \cite{BorealeS98}.}
 For instance, in a $\pi$-calculus extrusion transition $(y) \overline{x}y.p \trans{\overline{x}(z)} p[\nicefrac{z}{y}]$ we observe a free name $x$ and a fresh name $z$, which then becomes free in the continuation. Analogously, in a transition $O \rhd c \ctrans{K}{e}{a} \delta(O,K,e_a) \rhd c'$ the elements of $K$ are ``free'' events, in the sense that they occur in $c$, and $e$ is a fresh one, which is then added to the continuation. As in the $\pi$-calculus, event generation causes \clts{} to have infinitely-many states and to be infinitely-branching , because there are infinitely-many transitions and continuations from any state, differing only for the identity of the fresh event.

\begin{remark}
Even if initial markings are sets, firings may eventually produce a proper multiset, for instance when a transition puts a token in a place $s$ that is already marked. Instead, our causal markings are sets: they can never contain two occurrences of $K \caus s$, for any $K$. In fact, suppose the first of the described firings becomes a CG transition that goes to a P-marking including $K \caus s$. Then, since the second transition fires later, it will generate an event $e_a \notin K$ and a target P-marking that includes both $K \caus s$ and a new $K' \caus s$ such that $e_a \in K'$, so $K \neq K'$.
\end{remark}

\begin{example}
\autoref{fig:runex-cgc} depicts some transitions of the \clts{} for the running example. It shows only the reachable part from $\emptyset \rhd \emptyset \caus m_0$, up to a certain depth. Each state has three kinds of outgoing transitions, corresponding to the three net transitions. The figure only shows one transition for each kind, but there are actually infinitely many ones, one for each fresh event.
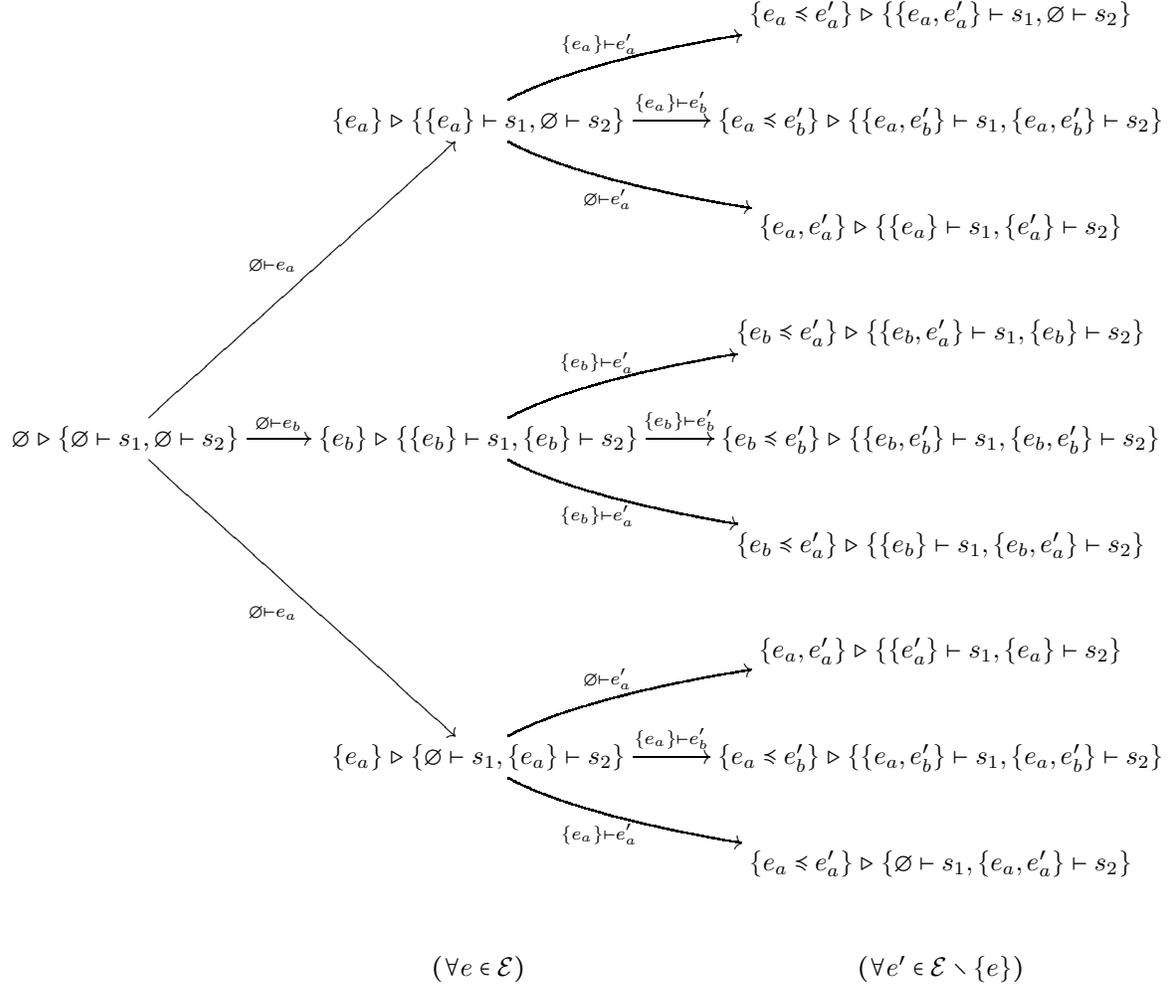
\begin{figure}[p]
\centering
\[
\xymatrix{
&& 
\{e_a \lss e'_a\} \rhd \{\{e_a,e'_a\} \caus s_1,\emptyset \caus s_2\}
\\
&
\{e_a\} \rhd \{\{e_a\}\caus s_1,\emptyset \caus s_2 \}
\ar@(ur,l)[ur]^-{\{e_a\} \caus e'_a}
\ar[r]^-{\{e_a\} \caus e'_b}
\ar@(dr,l)[dr]_-{\emptyset \caus e'_a}
&
\{e_a \lss e'_b\} \rhd \{\{e_a,e'_b\} \caus s_1,\{e_a,e'_b\} \caus s_2\}
\\
&&
\{e_a , e'_a\} \rhd \{\{e_a\} \caus s_1,\{e'_a\} \caus s_2\}
\\
&&
\{e_b \lss e'_a\} \rhd \{\{e_b,e'_a\}\caus s_1,\{e_b\} \caus s_2 \}
\\
\emptyset \rhd \{ \emptyset \caus s_1 , \emptyset \caus s_2 \}
\ar[uuur]^{\emptyset \caus e_a}
\ar[r]^-{\emptyset \caus e_b}
\ar[dddr]_{\emptyset \caus e_a}
& 
\{e_b\} \rhd \{\{e_b\}\caus s_1,\{e_b\} \caus s_2 \}
\ar@(ur,l)[ur]^-{\{e_b\} \caus e'_a}
\ar[r]^-{\{e_b\} \caus e'_b}
\ar@(dr,l)[dr]_-{\{e_b\} \caus e'_a}
&
\{e_b \lss e'_b\} \rhd \{\{e_b,e'_b\}\caus s_1,\{e_b,e'_b\} \caus s_2 \}
\\
&&
\{e_b \lss e'_a\} \rhd \{\{e_b\}\caus s_1,\{e_b,e'_a\} \caus s_2 \}
\\
&&
\{e_a , e'_a\} \rhd \{\{e'_a\}\caus s_1,\{e_a\} \caus s_2  \}
\\
&
\{e_a\} \rhd \{\emptyset\caus s_1,\{e_a\} \caus s_2 \}
\ar@(ur,l)[ur]^-{\emptyset \caus e'_a}
\ar[r]^-{\{e_a\} \caus e'_b}
\ar@(dr,l)[dr]_-{\{e_a\} \caus e'_a}
&
\{e_a \lss e'_b\} \rhd \{\{e_a,e'_b\} \caus s_1,\{e_a,e'_b\} \caus s_2 \}
\\
&&
\{e_a \lss e'_a\} \rhd \{\emptyset\caus s_1,\{e_a,e'_a\} \caus s_2 \}
\\
&
(\forall e \in \enames)
&
(\forall e' \in \enames \setminus \{e\} )
}
\]
\caption{\clts{} for the running example (initial fragment).}
\label{fig:runex-cgc}
\end{figure}
\end{example}
We now introduce bisimulations for \clts.
\begin{definition}[Concrete causal bisimulation]
\label{def:cbisim}
A \emph{concrete causal bisimulation} (\cb-bisimulation in short) is a family of relations $\{R_O\}$ on P-markings, indexed by $Act$-labeled posets, such that:
\begin{itemize}
    \item whenever $(O_1 \rhd c_1 , O_2 \rhd c_2) \in R_O$ then $O_1 = O_2 = O$;
    \item whenever $(O \rhd c_1 , O \rhd c_2) \in R_O$ and $O \rhd c_1 \ctrans{K}{e}{a} O' \rhd c_1'$ then $O \rhd c_2 \ctrans{K}{e}{a} O' \rhd c_2'$ and $(O' \rhd c_1',O' \rhd c_2') \in R_{O'}$ (and viceversa).
\end{itemize}
The concrete causal bisimilarity is the greatest such family and is denoted by $\cbisim$.
\end{definition}

\subsection{Abstract CG}
\label{ssec:aclts}

We now introduce an \emph{abstract CG}, where we only take posets up to isomorphism. We write $\isorep{O}$ for the isomorphism representative of $O$, and we call it \emph{abstract poset}. We call \emph{abstract} a P-marking of the form $\isorep{O} \rhd c$. 

Given an abstract poset $O$, $K \subseteq |O|$ and $a \in Act$, we assume the following operations:
\begin{itemize}     
    \item $\delta(O,K,a)$, generating $\isorep{\delta(O,K,e_a)}$, for any $e_a$; the actual identity of $e_a$ is not relevant, because of the quotient up to isomorphism;
    \item $new(O,K,a)$, giving the unique new event in $\delta(O,K,a)$;
    \item the morphism $old(O,K,a)$, embedding $O$ into $\delta(O,K,a)$;
\end{itemize}
These operations can be used to define the \emph{extension} of $\sigma \colon O \to O'$ (with $O,O'$ abstract posets) to a morphism $\ext{\sigma}_{K,a} \colon \delta(O,K,a) \to \delta(O',\sigma(K),a)$ given by
\[
	\ext{\sigma}_{K,a}(x) = 
	\begin{cases}
		new(O',\sigma(K),a) & x = new(O,K,a) \\
		old(O',\sigma(K),a)(\sigma(y)) & x = old(O,K,a)(y)
	\end{cases}
\]
The intuition is that $\ext{\sigma}_{K,a}$ does not mix up old and new events: it acts ``as'' $\sigma$ (modulo suitable embeddings) on events that were already in $O$, and maps the new event in $\delta(O,K,a)$ to the new one in $\delta(O',\sigma(K),a)$. To ease notation, we will just write $\ext{\sigma}$ when $K$ and $a$ are clear from the context.
\begin{example}
Suppose $O_1 = \{x_a,x'_b\}$ and $O_2 = \{y_a,y'_b,y''_c\}$ are discrete abstract posets, and let $\sigma \colon O_1 \to O_2$ map $x_a$ to $y_a$ and $x'_b$ to $y'_b$. Let $\hat{x}_z$ (resp.\ $\hat{y}_z$) be the image of $x_z$ via $old(O,\{x_a,x'_b\},d)$ (resp.\ via $old(O',\{y_a,y'_b\},d)$), for $z \in \{a,b\}$. Then we have 
\[
    \delta(O_1,\{x_a,x'_b\},d) =
    \begin{gathered}
        \xymatrix@C=0ex{
            & new(O_1,\{x_a,x'_b\},d)\\
            \hat{x}_a \ar[ur] && \hat{x}_b' \ar[ul]
        }
    \end{gathered}
    \qquad
    \delta(O_2,\{y_a,y_b'\},d) =
    \begin{gathered}
        \xymatrix@C=0ex{
            & new(O_2,\{y_a,y'_b\},d) \\
            \hat{y}_a \ar[ur] & \hat{y}'_b \ar[u] & \hat{y}''_c
        }
    \end{gathered}
\]
where arrows represent ordered pairs (reflexive pairs are omitted). Then $\ext{\sigma} \colon \delta(O_1,\{x_a,x_b'\},d) \to \delta(O_2,\{y_a,y_b'\},d)$ maps $\hat{x}_a$ to $\hat{y}_a$, $\hat{x}'_b$ to $\hat{y}'_b$ and $new(O,\{x_a,x'_b\},d)$ to $new(O_2,\{y_a,y'_b\},d)$.
\label{ex:aclts}
\end{example}
We now introduce the \emph{abstract CG}. Its states are abstract P-markings and its labels have the form $K \caus a$. Labels have the same meaning as in \clts, but here there is no need to observe the generated event: it will always be $new(O,K,a)$, if $O$ if the source P-marking's poset. 

In order to translate concrete P-markings, and their transitions, to their abstract counterparts in \aclts, we fix an \emph{abstraction isomorphism} $\alpha_O \colon O \to \isorep{O}$, for each poset $O$, giving a canonical representative of each event in $O$. In the following we write $\norm{x}{O}$ for the ``abstract version'' of $x$, namely $x \alpha_O$. We also introduce an operation $\norm{c}{O,K,e_a}$. It will be applied to causal markings $c$ appearing in continuations of transitions of \clts, namely those P-markings of the form $\delta(O,K,e_a) \rhd c$. Intuitively, given a transition in \clts, the operation $\norm{-}{O,K,e_a}$ applies the abstraction isomorphism of the source P-marking to its continuation, so that events of source and continuation are consistent with each other and the fresh event generated by the transition always becomes the canonical new one. Formally, $\norm{c}{O,K,e_a}$ is defined as follows: events in $O$ are mapped via $\alpha_O$ and then embedded into $\isorep{\delta(O,K,e_a)}$ via $old(\isorep{O},\norm{K}{O},a)$  (notice that $\isorep{\delta(O,K,e_a)}=\delta(\isorep{O},\norm{K}{O},a)$, because they are isomorphic); and $e_a$ is embedded into $\isorep{\delta(O,K,e_a)}$ as $new(\isorep{O},\norm{K}{O},a)$. 

\begin{definition}[abstract CG]
\label{def:aclts}
The \emph{abstract CG} (\aclts) is the smallest CG generated by the following rule
\[
    \frac
    {
        O \rhd c \ctrans{K}{e}{a} \delta(O,K,e_a) \rhd c'
    }
    {
        \isorep{O} \rhd \norm{c}{O} \actrans{\norm{K}{O}}{a} \delta(\isorep{O},\norm{K}{O},a) \rhd \norm{c'}{O,K,e_a}
    }
\]
\end{definition}

The most important fact to notice is that \aclts{} is finitely branching. In fact, even if there are infinitely-many concrete P-markings that generate the transitions of an abstract P-marking $O \rhd c$, they are all isomorphic. To see this, take any two P-markings $O_1 \rhd c_1$ and $O_2 \rhd c_2$ such that $\norm{c_1}{O_1} = \norm{c_2}{O_2} = c$. Then we have $c = c_1\alpha_{O_1}^{-1} = c_2\alpha_{O_2}^{-1}$, so $c_2 = c_1\sigma$, where $\sigma$ is the isomorphism $\alpha_{O_2}^{-1} \circ \alpha_{O_1}$. The following lemma states the correspondence between transitions of such P-markings.
\begin{lemma}
Let $\sigma \colon O_1 \to O_2$ be an isomorphism. Then $O_1 \rhd c_1 \ctrans{K}{e}{a} \delta(O_1,K,e_a) \rhd c_1'$ if and only if $O_2 \rhd c_1 \sigma \ctrans{\sigma(K)}{e'}{a} \delta(O_2,\sigma(K),e'_a) \rhd c_1'\funext{\sigma}{e'_a}{e_a}$, for any $e' \notin X_{O_2}$.
\label{lem:ctrans-iso}
\end{lemma}
If we take any two transitions of $O_1 \rhd c_1$ and $O_2 \rhd c_2$ that correspond by this lemma, and we apply the rule in \autoref{def:aclts} to them, it can be easily verified that we get the same transition, no matter the choice of $e_a$ and $e'_a$.
Therefore, all the infinitely-many P-markings whose abstract version is $O \rhd c$ generate precisely the same transitions of $O \rhd c$, and transitions that differ for the choice of the fresh event are all identified. This means that \aclts{} is finitely-branching.

There is again a similarity with the $\pi$-calculus. A well-known technique to make the $\pi$-calculus LTS finitely-branching is to only take $\alpha$-equivalence representatives. For instance, if $(y)\overline{x}y.p$ is such a representative, then the transition $(y)\overline{x}y.p\trans{\overline{x}(y)}p$ is enough to represent all the analogous transitions from $\alpha$-equivalent processes. We can also omit $y$ from the label, because its identity uniquely depends on the free names of $(y)\overline{x}y.p$. This is similar to the presentation of the $\pi$-calculus using abstraction and concretion operators \cite[4.3.1]{SangiorgiW01}. Here a transition from $(y) \overline{x}y.p$ is labeled by $\overline{x}$ and goes to the concretion $\langle \nu y \rangle p$, where $y$ is bound. Incidentally, this presentation naturally arises from the coalgebraic semantics of the $\pi$-calculus \cite{FioreT01}, and its implementation in logical frameworks.

\begin{example}
\label{ex:running-aclts}
The \aclts{} for the running example can be represented again by \autoref{fig:runex-cgc}. If we assume that depicted posets are abstract (i.e., translation maps from concrete to abstract posets are identities) then, in order to get a \aclts, we just have to remove the universal quantification over events, and also remove the generated event from the label. The result is a finitely-branching CG, where each state has only one transition for each net transition. The state-space is still infinite, because posets keep growing along transitions.
\end{example}
\begin{definition}[Abstract causal bisimilarity]
\label{def:acbisim}
An \emph{abstract causal bisimulation} (\acb-bisimulation in short) is a family of relations $\{ R_O \}$, indexed by abstract posets, such that:
\begin{itemize}
    \item whenever $(O_1 \rhd c_1, O_2 \rhd c_2) \in R_O$ then $O_1 = O_2 = O$;
    \item whenever $(O \rhd c_1, O \rhd c_2) \in R_O$ and $O \rhd c_1 \actrans{K}{a} O' \rhd c_1'$ then $O \rhd c_2 \actrans{K}{a} O' \rhd c_2'$ and $(O' \rhd c_1',O' \rhd c_2') \in R_{O'}$ (and viceversa).
\end{itemize}
The greatest such relation is denoted by $\acbisim$.
\end{definition}
We have the following correspondence between $\cbisim$ and $\acbisim$.
\begin{theorem}
Let $O \rhd c_1$ and $O \rhd c_2$ be (concrete) P-markings. Then $O \rhd c_1 \cbisim O \rhd c_2$ if and only if $\isorep{O} \rhd \norm{c_1}{O} \acbisim \isorep{O} \rhd \norm{c_2}{O}$.
\label{thm:acbisim-cbisim}
\end{theorem}
We list some closure properties, which will be important in the following. 
\begin{proposition}
\label{prop:aclts-clos}
Transitions of \aclts{} are \emph{preserved} and \emph{reflected} by order-embeddings $\sigma \colon O \to O'$, that is:
\begin{enumerate}[$(i)$]
	\item If $O \rhd c \actrans{K}{a} \delta(O,K,a) \rhd c'$ then $O' \rhd \dclos{(c \sigma)}{O'} \actrans{\sigma(K)}{a} \delta(O',\sigma(K),a) \rhd \dclos{(c'\ext{\sigma})}{\delta(O',\sigma(K),a)}$ (preservation);
	\item If $O' \rhd \dclos{(c \sigma)}{O'} \actrans{K'}{a} \delta(O',K',a) \rhd c'$ then there are $K$ and $c''$ such that $\sigma(K) = K'$, $\dclos{(c'' \ext{\sigma})}{\delta(O',K',a)} = c'$ and $O \rhd c \actrans{K}{a} \delta(O,K,a) \rhd c''$ (reflection).
\end{enumerate}
\end{proposition}
The definition of preservation and reflection are quite involved, due to the presence of event generation and the need of applying the closure operator to compute proper continuations. We will see that the categorical counterparts of these properties will be remarkably simpler.

\begin{example}
\label{ex:tr-pres}
We motivate the requirement of order-reflection by showing that transitions of \aclts{} are not reflected by functions without such property. 

Consider the marked net of \autoref{fig:exnet}.
\begin{figure}[t]
\begin{center}
\begin{tikzpicture}[auto]    
    \node[place,label=$s_1$,tokens=1] (s1) {};
    \node[place,label=$r_1$,right= of s1] (r1) {};
    \node[place,label=$s_2$,below= of s1,tokens=1] (s2) {};
    \node[place,label=$r_2$,right= of s2] (r2) {};
    
    \coordinate (mid) at ($(r1)!0.5!(r2)$);
    \node[place,label=$s_3$,right= 2cm of mid] (s3) {};

    \node[transition] (t1) at ($(s1)!0.5!(r1)$) {$a$}
        edge [pre] node {} (s1)
        edge [post] node {} (r1);
    
    \node[transition] (t2) at ($(s2)!0.5!(r2)$) {$b$}
        edge [pre] node {} (s2)
        edge [post] node {} (r2);

    \node[transition,right= of mid] (t3) {$c$}
        edge [pre] node {} (r1)
        edge [pre] node {} (r2)
        edge [post] node {} (s3);

\end{tikzpicture}
\end{center}
\caption{Example net.}
\label{fig:exnet}
\end{figure}
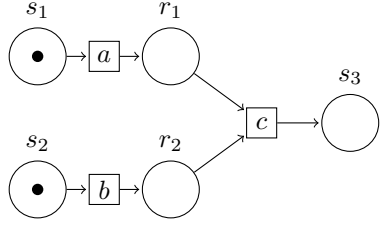
We can derive its \aclts{} as shown for the running example. In it, from the initial P-marking $\emptyset \rhd \{ \emptyset \caus s_1, \emptyset \caus s_2 \}$ we can reach the transition
\[
	\{e_a,e'_b\} \rhd \{ \{e_a\} \caus r_1,\{e'_b\} \caus r_2 \} \actrans{\{e_a,e'_b\}}{c} 
	\{e_a \lss e''_c,e'_b \lss e''_c \} \rhd
	\{ \{e_a,e'_b,e''_c\} \caus s_3 \}
\]
Consider the function $\sigma \colon \{e_a,e'_b\} \to \{e_a \lss e'_b\}$, mapping events to themselves. Clearly $\sigma$ does not reflect posets.
If we apply $\sigma$ and then $\dclos{}{\{e_a \lss e'_b\}}$ to the source P-marking we get
\[
    \{e_a \lss e'_b\} \rhd \{ \{e_a\} \caus r_1,\{e_a,e'_b\} \caus r_2 \}    
\]
but its $c$ transition is
\[
    \{e_a \lss e'_b\} \rhd \{ \{e_a\} \caus r_1,\{e_a,e'_b\} \caus r_2 \}  \actrans{\{e'_b\}}{c} \{e_a \lss e''_c \lss e'_b \} \rhd \{ \{e_a,e'_b,e''_c\} \caus s_3 \}  
\]
because only $e'_b$ is maximal. However, this transition cannot be obtained from the one of $\{e_a,e'_b\} \rhd \{ \{e_a\} \caus r_1,\{e'_b\} \caus r_2 \}$ via an application of $\sigma$.
\end{example}
The following theorem is a consequence of \autoref{prop:aclts-clos}.
\begin{theorem}
\label{thm:acbisim-clos}
$\acbisim$ is closed under order-embeddings. Explicitly: for all order-embeddings $\sigma \colon O \to O'$, we have $O \rhd c \acbisim O \rhd c'$ if and only if $O' \rhd \dclos{(c \sigma)}{O'} \acbisim O' \rhd \dclos{(c' \sigma)}{O'}$.
\end{theorem}

\subsection{Immediate causes CG}

We now introduce a further refinement of \aclts{}, called \emph{immediate causes CG} (\iclts): we keep only \emph{immediate causes}, i.e., causes that are maximal w.r.t.\ at least one of the tokens, and we identify isomorphic states. Immediate causes of a causal marking w.r.t.\ a poset $O$ are given by
\[
	ic_O(K \caus s) = max_O(K) \qqquad
	ic_O(c_1 \cup c_2) = ic_O(c_1) \cup ic_O(c_2) 
\]
We define isomorphism of P-markings as follows: $O \rhd c \cong O' \rhd c'$ if and only if there is an isomorphism $\sigma \colon O \to O'$ such that $c \sigma = c'$. We denote by $\isorep{O \rhd c}$ a chosen representative for the isomorphism class of $O \rhd c$.
\begin{definition}[Minimal P-marking]
\label{def:min-pmark}
A \emph{minimal} P-marking $O \RHD c$ is an abstract P-marking such that:
\begin{itemize}
    \item $|O| = \causes(c)$;
    \item for each $K \caus s \in c$, $K \subseteq ic_O(c)$;
    \item it is a canonical isomorphism representative, i.e., $O \RHD c = \isorep{O \rhd c}$.
\end{itemize}
\end{definition}
Consider an abstract P-marking $O \rhd c$. In order to compute the corresponding minimal P-marking $\minp{O \rhd c}$, we first take immediate causes for each token. Then, since the resulting P-marking may not be abstract, we take its canonical isomorphism representative. Formally, let $\imm{O}$ be $O$ restricted to ${ic_O(c)}$, then
\[
	\minp{O \rhd c} = [\imm{O} \rhd norm_{\imm{O}}(c)]_{\cong}
\]
where $norm_O(K \caus s) =  K \cap |\imm{O}| \caus s$ and has an element-wise action on sets. We denote by $\minpfun{O \rhd c}$ the map $\isorep{\imm{O}} \to O$ obtained by composing a chosen isomorphism $\isorep{\imm{O}} \to \imm{O}$ and the embedding $\imm{O} \subto O$. 
\begin{definition}[Immediate causes CG]
\label{def:iclts}
The \emph{immediate causes CG} (\iclts) is the smallest CG generated by the following rule
\[
	\frac
	{
		O \rhd c \actrans{K}{a} O' \rhd c' 
	}
	{
		O \RHD c \ictrans{K}{a}{\minpfun{O' \rhd c'}}
		\minp{O' \rhd c'}
	}
\]
\end{definition}
This rule relies on the fact that minimal P-markings are also ordinary ones, so it takes the transition in \aclts{} from a minimal P-marking, replaces the continuation $O' \rhd c'$ with its minimal version $\minp{O' \rhd c'}$ and, in order to keep track of the original identity of events, equips the transition with a \emph{history map} $\minpfun{O' \rhd c'}$, mapping canonical events to the original ones. In particular, the one with image $new(O',K,a)$ is the fresh event generated by the original transition. 

The \iclts{} has a finite state-space in many cases. We give a sufficient condition on the net from which the \iclts{} is generated.

\begin{proposition}
Given a net $N$ with initial marking $m_0$, if $\fires{m_0}$ is finite then the corresponding \iclts, reachable from $\emptyset \RHD \emptyset \caus m_0$, has a finite state-space.
\label{prop:reach-size}
\end{proposition}
\begin{example}
\label{ex:running-iclts}
In order to derive a \iclts{} for the running example, we take the P-markings of \autoref{fig:runex-cgc} and we compute their minimal versions. For instance, we have
\[
    \xymatrix{
    \{e_b \lss e'_b\} \rhd \{\{e_b , e'_b\} \caus s_1,\{e_b,e_b'\} \caus s_2 \}
    \ar[d]^{\text{immediate causes}}
    \\
    \{e'_b\} \rhd \{\{e'_b\} \caus s_1, \{e'_b\} \caus s_2 
    \}
    \ar[d]^{\text{canonical representative}}
    \\
    \{e_b\} \RHD \{\{e_b\} \caus s_1,\{e_b\} \caus s_2 
    \}    
    }
\]
because we assumed that $\{e_b\}$ is an abstract poset.
Notice that the resulting P-marking is already in \autoref{fig:runex-cgc}. This is a crucial fact: minimization identifies many states and in some cases it even produces a finite state-space, as stated in \autoref{prop:reach-size}.
This is indeed the case for the running example.

\autoref{fig:runex-iccg} shows the part of the running example's \iclts{} that is reachable from  $\{e_b\} \RHD \{\{e_b\} \caus s_1,\{e_b\} \caus s_2\}$. Most history maps are irrelevant, so they are omitted. Notice that in the \aclts{}, from this P-marking, there are infinitely many transitions with action $b$. These all become a single loop over the same P-marking in the \iclts; the associated history map $h_1$ tells that $e_b$, after the transition, represents the most recent event, and that the previous event is discarded. Analogously for the two loops over $\{e_a,e'_a\} \RHD \{\{e_a\} \caus s_1 \{e'_a\} \caus s_2\}$. The interesting fact to notice is that our definition of $h_2$ and $h_3$ is not the only possible one. For instance, we could exchange the images of $e_a$ and $e'_a$ in the definition of $h_2$. This is due to the fact that $\{e_a,e'_a\}$ has an automorphism that swaps $e_a$ and $e'_a$. 
\begin{figure}
\centering
\[
\xymatrix@R+3ex@C=0ex{
&&
\{e_b \lss e'_a\} \RHD \{\{e_b,e'_a\}\caus s_1,\{e_b\} \caus s_2 \}
\icar[dr]^-{\{e_b\} \caus a}
\icar@/^1pc/[dl]_-{\{e_b\} \caus b}
\icar@(ul,ur)^{\{e'_a\} \caus a}
\\
& 
\{e_b\} \RHD \{\{e_b\}\caus s_1,\{e_b\} \caus s_2 \}
\icar@/^1.5pc/[ur]^-{\{e_b\} \caus a}
\save 
!<-20pt,0pt> \icar@(ul,ur)^{\{e_b\} \caus b}_{h_1} 
\restore
\icar@/_1.5pc/[dr]^-{\{e_b\} \caus a}
&&
\{e_a , e'_a\} \RHD \{\{e_a\} \caus s_1,\{e'_a\} \caus s_2\}
\icar[ll]_-{\{e_a,e_a'\} \caus b}
\icar@(ul,ur)^{\{e_a\} \caus a}_{h_2}
\icar@(dl,dr)_{\{e'_a\} \caus a}^{h_3}
\\
&&
\{e_b \lss e'_a\} \RHD \{\{e_b\}\caus s_1,\{e_b,e'_a\} \caus s_2  \}
\icar@(dl,dr)_{\{e'_a\} \caus a}
\icar[ur]^{\{e_b\} \caus a}
\icar@/_1pc/[ul]_-{\{e'_a\} \caus b}
}
\]

\begin{align*}
h_1 \colon \{e_b\} &\to \{ e_b \lss e'_b \} 
&
h_2 \colon \{e_a,e'_a\} &\to \{ e_a \lss e''_a , e'_a \}
&
h_3 \colon \{e_a,e'_a\} &\to \{ e_a , e'_a \lss e''_a \}
\\
e_b &\mapsto e'_b
&
e_a &\mapsto e''_a
&
e_a &\mapsto e_a
\\
&&
e'_a &\mapsto e'_a
&
e'_a &\mapsto e''_a
\end{align*}

\caption{\iclts{} for the running example.}
\label{fig:runex-iccg}
\end{figure}
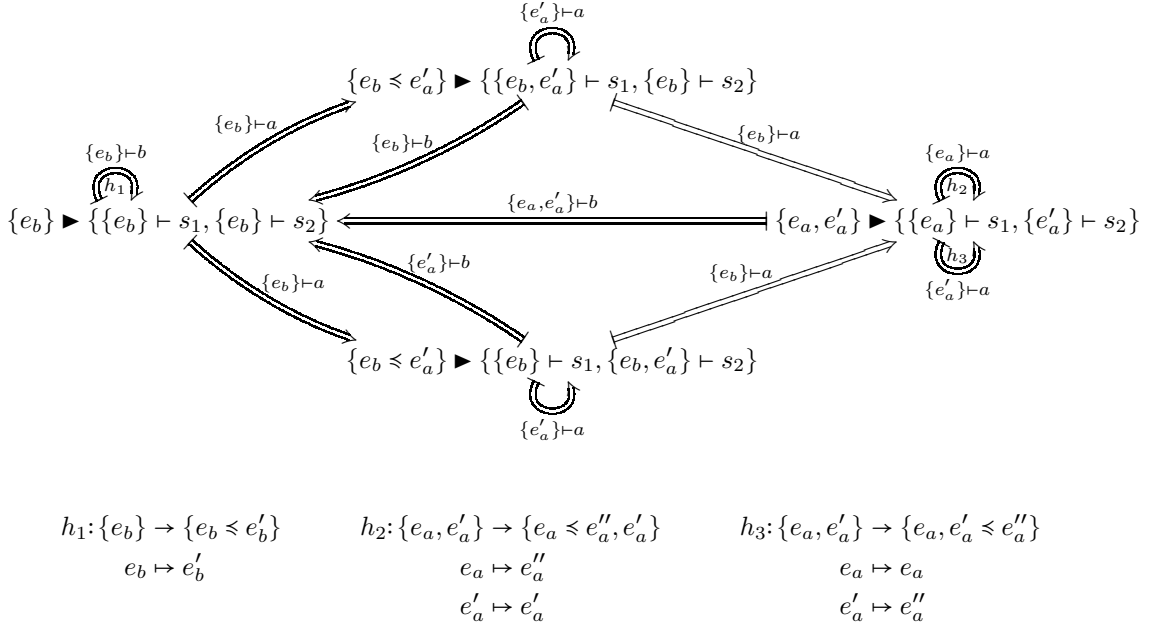
\end{example}
\begin{remark}
The generation of the \iclts{} from a net has been performed in two steps for the sake of clarity, but we can easily imagine an algorithm that performs it in a single step and incrementally. Given any P-marking, this is turned into a minimal one by taking immediate causes and then its canonical representative. Then outgoing transitions are computed from this P-marking, and the algorithm is applied to their continuations. Notice that minimizing a P-marking may yield a previously computed one: in this case the algorithm is not reapplied on that P-marking.
\end{remark}
The notion of bisimilarity for \iclts{} is more involved: 
while, given two P-markings, we may find a common poset for them (if any), which enables them to be compared w.r.t.\ $\acbisim$, this is not always possible for posets of minimal P-markings. In other words, events in ordinary P-markings have a \emph{global} identity, while those in minimal P-markings have a \emph{local} identity. Therefore, we need to introduce an explicit correspondence between them. This correspondence can be a partial function, because some events may not be observable.
\begin{definition}[Immediate causes bisimilarity]
\label{def:icbisim}
An \emph{immediate causes bisimulation} $R$ (\icb-bisimulation in short) is a ternary relation such that, whenever $(O_1 \RHD c_1,\sigma,O_2 \RHD c_2) \in R$:
\begin{itemize} 
    \item $\sigma$ is a partial isomorphism (i.e., an isomorphism between subposets) from $O_1$ to $O_2$;
	\item if $O_1 \RHD c_1 \ictrans{K}{a}{h_1} O_1' \RHD c_1'$ then $\sigma$ is defined on $K$, and there are $O_2 \RHD c_2 \ictrans{\sigma(K)}{a}{h_2} O_2' \RHD c_2'$ and $\sigma'$ such that $(O_1' \RHD c_1',\sigma',O_2' \RHD c_2') \in R$ and the following diagram commutes
\[
    \xymatrix{
        O_1' \ar[r]^-{h_1} \ar[d]_{\sigma'} & \delta(O_1,K,a) \ar[d]^{\ext{\sigma}} \\
        O_2' \ar[r]_-{h_2} & \delta(O_2,\sigma(K),a) \\
    }
\]
	\item if $O_2 \RHD c_2 \ictrans{K}{a}{h_2} O_2' \RHD c_2'$ then $\sigma$ is defined on $K$, and there are $O_1 \RHD c_1 \ictrans{\sigma(K)^{-1}}{a}{h_1} O_1' \RHD c_1'$ and $\sigma'$ as in the previous item.
\end{itemize}
The greatest such bisimulation is denoted $\icbisim$. We write $O_1 \RHD c_1 \icbisim^\sigma O_2 \RHD c_2$ to mean 
\[
	(O_1 \RHD c_1,\sigma, O_2 \RHD c_2) \in \icbisim.
\]
\end{definition}
The commuting diagram essentially says that $\sigma'$ should never map old events to new ones (or viceversa). More precisely, given $x \in |O_1'|$, we have two cases:
\begin{itemize}
    \item $h_1(x) = new(O_1,K,a)$, then, by definition, $h_1(x)$ is mapped by $\ext{\sigma}$ to $new(O_2,\sigma(K),a)$, so $\sigma'(x) = y$ such that $h_2(y) = new(O_2,\sigma(K),a)$;
    \item $h_1(x) = old(O_1,K,a)(x')$, for some $x'$, then $h_1(x)$ is mapped by $\ext{\sigma}$ to $old(O_2,\sigma(K),a)(\sigma(x'))$, so $\sigma'(x) = y$ such that $h_2(y) = old(O_2,\sigma(K),a)(\sigma(x'))$.
\end{itemize}
We have the following correspondence between $\icbisim$ and $\acbisim$.
\begin{theorem}
\label{thm:corr-ic-ac}
$\icbisim$ is fully abstract w.r.t.\ $\acbisim$ in the following sense:
\begin{enumerate}[$(i)$]
	\item If $O \rhd c_1 \acbisim O \rhd c_2$ then $\minp{O \rhd c_1} \icbisim \minp{O \rhd c_2}$;
	\label{ac2ic}
	\item If $O_1 \RHD c_1 \icbisim^\sigma O_2 \RHD c_2$ then for all $O \rhd \hat{c}_1$ and $O \rhd \hat{c}_2$ such that:
	\begin{enumerate}
		\item $\minp{O \rhd \hat{c}_1} = O_1 \RHD c_1$ and $\minp{O \rhd \hat{c}_2} = O_2 \RHD c_2$;
		\item $\minpfun{O \rhd \hat{c}_1}|_{dom(\sigma)} = \minpfun{O \rhd \hat{c}_2} \circ \sigma$;
	\end{enumerate}
	we have $O \rhd \hat{c}_1 \acbisim O \rhd \hat{c}_2$.
	\label{ic2ac}
\end{enumerate}
\end{theorem}
Statement \ref{ac2ic} is self-explanatory. Statement \ref{ic2ac} says that if we have two equivalent minimal P-markings $O_1 \RHD c_1 \icbisim^\sigma O_2 \RHD c_2$ and we take any two P-markings $O \rhd \hat{c}_1$ and $O \rhd \hat{c}_2$ whose minimal versions are $O_1 \RHD c_1$ and $O_2 \RHD c_2$ respectively ((ii)(a)), these are equivalent provided that local events matched by $\sigma$ have the same global interpretation as events of $O$ ((ii)(b)).

\subsection{Immediate causes CG with symmetries}
\label{ssec:icslts}

The final step is to introduce \emph{symmetries} over states of CG. Given an abstract poset $O$, a \emph{symmetry over $O$} is a set $\Phi$ of automorphisms $O \to O$ (called just \emph{permutations} hereafter) such that $id \in \Phi$ and it is closed under composition. This section is an adaptation of the work in \cite{Pistore99,MontanariP05} on the set-theoretic version of HD-automata for the $\pi$-calculus.

We now motivate the introduction of symmetries. We say that two \iclts s are \emph{isomorphic} when there is a bijective correspondence $\omega$ between their P-markings and, for each P-marking $O \RHD c$ of the former such that $\omega(O \RHD c) = O' \RHD c'$, transitions from $O' \RHD c'$ can be obtained from those of $O \RHD c$ via an isomorphism. In the case of ordinary labeled transition systems (LTSs), one can compute minimal versions w.r.t. bisimilarity, where all bisimilar states have been identified. Bisimilar LTSs have isomorphic minimal versions, so we may use any of them as canonical representative of the class of bisimilar LTSs. This cannot be done for \iclts s, because of the following fact.
\begin{proposition}
There are minimal \iclts s that are $\icbisim$-bisimilar but not isomorphic.
\end{proposition}
\begin{example}
Consider the P-marking $\{e_a,e'_a\} \RHD \{\{e_a\} \caus s_1,\{e'_a\} \caus s_2 \}$ of \autoref{ex:running-iclts} and its looping transitions. Take another P-marking $\{e_a,e'_a\} \RHD \{\{e_a\} \caus s_1', \{e'_a\} \caus s_2' \}$ with the following transitions
\[
    \xymatrix{
        \{e_a,e'_a\} \RHD \{\{e_a\} \caus s_1', \{e'_a\} \caus s_2' \}
        \icar@(ul,ur)^{\{e_a\} \caus a}_{h_4}
        \icar@(dl,dr)_{\{e'_a\} \caus a}^{h_5}
    }
    \quad
    \begin{array}{rlrl}
         h_4 \colon \{e_a,e'_a\} \hspace{-1.5ex}&\mapsto \{e_a \lss e''_a,e'_a \} 
         & 
         \quad
         h_5 \colon \{e_a,e'_a\} \hspace{-1.5ex}&\mapsto \{e_a , e'_a \lss e''_a\}
         \\
         e_a &\mapsto e''_a
         &
         e_a &\mapsto e''_a
         \\
         e'_a &\mapsto e'_a
         &
         e'_a &\mapsto e_a
    \end{array}
\]
Notice that we have $h_4 = h_2$ and $h_5 = h_3 \circ \phi$, where $\phi$ switches $e_a$ and $e'_a$. 

Suppose we want to find a minimal realization of these CGs. They are not isomorphic, in the sense that there is no permutation on $\{e_a,e'_a\}$ that, applied to labels and composed with history maps, turns transitions of the former CG into those of the latter. However, we have
\[
    \{e_a,e'_a\} \RHD \{\{e_a\} \caus s_1,\{e'_a\} \caus s_2 \} \icbisim^{\phi} \{e_a,e'_a\} \RHD \{\{e_a\} \caus s_1', \{e'_a\} \caus s_2' \} \enspace ,
\]
so these states should be identified in some way. This way is provided by symmetries: minimal behavior, according to $\icsbisim$, is invariant under $\phi$, so we can identify those P-markings, provided that the resulting state is annotated with $\phi$ and possibly other permutations that fix the state.

The same argument applies when considering versions of the same \iclts{} that only differ for the choice of history maps: if $s_1' = s_1$ and $s_2' = s_2$ in the P-marking $\{e_a,e'_a\} \RHD \{\{e_a\} \caus s_1', \{e'_a\} \caus s_2' \}$ above, then the P-marking $\{e_a,e'_a\} \RHD \{\{e_a\} \caus s_1', \{e'_a\} \caus s_2' \}$ is bisimilar to itself under the permutation $\phi$. This has a practical consequence: when constructing the \iclts{} for a given net, one should not spend computational effort in computing the ``right'' history maps, because the choice of history maps does not affect bisimilarity and thus minimal models.
\end{example}

\begin{definition}[Minimal P-marking with symmetry]
A \emph{minimal P-marking with symmetry} is a triple $\psym{O}{\Phi}{c}$, where $O \RHD c$ is a minimal P-marking and $\Phi$ is a symmetry over $O$ such that $c \phi = c$, for all $\phi \in \Phi$.
\end{definition}
Symmetries allow us to remove some transitions from \iclts: we can only take one representative transition among all the \emph{symmetric} ones, i.e., those whose observable causes and history maps only differ for some permutations in the symmetries of source and target states. 
\begin{definition}[Symmetric transitions]
Given $\psym{O}{\Phi}{c}$, $\psym{O'}{\Phi'}{c'}$ and two transitions
\[
    O \RHD c \ictrans{K_1}{a}{h_1} O' \RHD c'
    \qquad 
    O \RHD c \ictrans{K_2}{a}{h_2} O' \RHD c'
\]
they are \emph{symmetric} if and only if there are $\phi \in \Phi$ and $\phi' \in \Phi'$ such that $K_2 = \phi(K_1)$ and the following diagram commutes
\[
    \xymatrix{
        O' \ar[r]^-{h_1} \ar[d]_{\phi'} & \delta(O,K_1,a)\ar[d]^{\ext{\phi}} \\
        O' \ar[r]_-{h_2} & \delta(O,K_2,a)
    }  
\]
We write $\symrep{K}$ and $\symrep{h}$ for a canonical choice of $K$ and $h$ among those of all the symmetric transitions. Actually $\symrep{-}$ depends on the considered symmetries $\Phi$ and $\Phi'$, but they are omitted to simplify notation: they will always be clear from the context.
\label{def:sym-tran}
\end{definition}
\begin{definition}[\icslts]
The \emph{\iclts{} with symmetries} (\icslts) is the smallest CG generated by the following rule
\[
    \frac
    {
        O \RHD c \ictrans{K}{a}{h} O' \RHD c'
    }
    {
        \psym{O}{\Phi}{c} \icstrans{\symrep{K}}{a}{\symrep{h}} \psym{O'}{\Phi'}{c'}
    }
\]
\label{def:icslts}
\end{definition}
The notion of bisimulation is analogous to \icb-bisimulation. However, P-markings are required to simulate each other only up to symmetries. More specifically, when comparing $\psym{O_1}{\Phi_1}{c_1}$ and $\psym{O_2}{\Phi_2}{c_2}$ under a mediating map $\sigma$, for each permutation in $\Phi_1$ and each transition of the first P-marking, we have to find a permutation in $\Phi_2$ and a transition of the second P-marking. The correspondence between observable causes and between history maps must be as in \icb-bisimulations, but the action of mediating maps is changed according to the considered permutations.
\begin{definition}[Immediate causes bisimulation with symmetries]
An \emph{immediate causes bisimulation with symmetries} $R$ (\icsb-bisimulation in short) is a ternary relation such that, whenever $(\psym{O_1}{\Phi_1}{c_1},\sigma,\psym{O_2}{\Phi_2}{c_2}) \in R$:
\begin{itemize} 
    \item $\sigma$ is a partial isomorphism from $O_1$ to $O_2$;
	\item for each $\phi_1 \in \Phi_1$ and $\psym{O_1}{\Phi_1}{c_1} \icstrans{K_1}{a}{h_1} \psym{O_1'}{\Phi_1'}{c_1}'$, $\sigma$ is defined on $\phi_1(K)$ and there are $\phi_2 \in \Phi_2$ and $\psym{O_2}{\Phi_2}{c_2} \icstrans{K_2}{a}{h_2} \psym{O_2'}{\Phi_2'}{c_2}'$ such that:
	\begin{itemize}
	    \item $K_2 = \gamma(K_1)$, for $\gamma = \phi_2^{-1} \circ \sigma \circ \phi_1$;
	    \item there is $\sigma'$ such that $(\psym{O_1'}{\Phi_1'}{c_1'},\sigma',\psym{O_2'}{\Phi_2'}{c_2'}) \in R$ and the following diagram commutes
\[
    \xymatrix{
        O_1' \ar[r]^-{h_1} \ar[d]_{\sigma'} & \delta(O_1,K_1,a) \ar[d]^{\ext{\gamma}} \\
        O_2' \ar[r]_-{h_2} & \delta(O_2,K_2,a) \\
    }
\]
\end{itemize}
(and viceversa)
\end{itemize}
The greatest such relation is denoted $\icsbisim$ and we write $\psym{O_1}{\Phi_1}{c_1} \icsbisim^{\sigma} \psym{O_2}{\Phi_2}{c_2}$ whenever $(\psym{O_1}{\Phi_1}{c_1},\sigma,\psym{O_2}{\Phi_2}{c_2}) \in \icsbisim$.
\end{definition}
As mentioned, symmetries allow computing minimal realizations, where all bisimilar P-markings are identified. More precisely, we can identify $\icsbisim$-equivalent P-markings, namely $\psym{O_1}{\Phi_1}{c_1}$ and $\psym{O_2}{\Phi_2}{c_2}$ that are related by $\icsbisim^{\sigma}$, for some $\sigma$. Then $\sigma$ becomes part of the state symmetry. Actually, $\sigma$ is a permutation between subposets of $O_1$ and $O_2$, but it can be shown that all $\icsbisim$-equivalent P-markings have the same poset of observable events on which $\sigma$ is defined. This means that $\sigma$ is indeed a permutation on that poset. 
\begin{definition}[Minimal \icslts]
The minimal \icslts{} is defined as follows:
\begin{itemize}
    \item states are canonical representatives of $\icsbisim$-equivalence, namely $\psym{O}{\Phi}{c}$ such that $\Phi = \{ \sigma \mid \exists \Phi' : \psym{O}{\Phi'}{c} \icsbisim^\sigma \psym{O}{\Phi'}{c} \}$;
    \item transitions are derived according to \autoref{def:icslts}.
\end{itemize}
\end{definition}
In order to compute the symmetry $\Phi$ of a canonical representative $\psym{O}{\Phi}{c}$, we take P-markings of the form $\psym{O}{\Phi'}{c}$ and we consider triples where $\psym{O}{\Phi'}{c}$ is bisimilar to itself. Notice that $\Phi$ may be different than $\Phi'$: some $\phi \in \Phi$, in fact, may not act as the identity of $c$; with a little abuse of notation, $\psym{O}{\Phi}{c}$ stands for a P-marking where every $\phi \in \Phi$ has identical action on $c$ up to bisimilarity. It can be proved that we do not need to consider non-canonical P-markings for the computation of $\Phi$ (see, e.g., \cite[5.2]{MontanariP05}).
\begin{example}
Consider the \iclts{} of \autoref{ex:running-iclts}. It can be regarded as a \icslts{} where all states have the singleton symmetry $\{id\}$. Its minimal version is depicted in \autoref{fig:runex-min}.
 Notice that the P-marking $\psym{\{e_a , e'_a\}}{\Phi_3}{\{\{e_a\} \caus s_1,\{e'_a\} \caus s_2\}}$ has a non-trivial symmetry, because we have $\psym{\{e_a , e'_a\}}{\{id\}}{\{\{e_a\} \caus s_1,\{e'_a\} \caus s_2\}} \icsbisim^{(e_a \ e'_a)} \psym{\{e_a , e'_a\}}{\{id\}}{\{\{e_a\} \caus s_1,\{e'_a\} \caus s_2\}}$.
\begin{figure}[t]
\centering
\[
\xymatrix@R+3ex@C=0ex{
\psym{\{e_b\}}{\Phi_1}{\{\{e_b\}\caus s_1,\{e_b\} \caus s_2 \}}
\icsar@(ul,ur)^{\{e_b\} \caus b}
\icsar@/^1pc/[dr]^-{\{e_b\} \caus a}
&&
\psym{\{e_a , e'_a\}}{\Phi_3}{\{\{e_a\} \caus s_1,\{e'_a\} \caus s_2\}}
\icsar[ll]_-{\{e_a,e_a'\} \caus b}
\icsar@(ul,ur)^{\{e_a\} \caus a}
\icsar@(dl,dr)_{\{e'_a\} \caus a}
\\
&
\psym{\{e_b \lss e'_a\}}{\Phi_2}{\{\{e_b,e'_a\}\caus s_1,\{e_b\} \caus s_2 \}}
\icsar@/^1pc/[ul]^-{\{e'_a\} \caus b}
\icsar@(dl,dr)_{\{e'_a\} \caus a}
\icsar[ur]^{\{e_b\} \caus a}
}
\]
\[
    \Phi_1 = \Phi_2 = \{id\} \qquad \Phi_3 = \{ (e_a \ e'_a) , id \}
\]
\caption{Minimal \icslts{} for the running example.}
\label{fig:runex-min}
\end{figure}
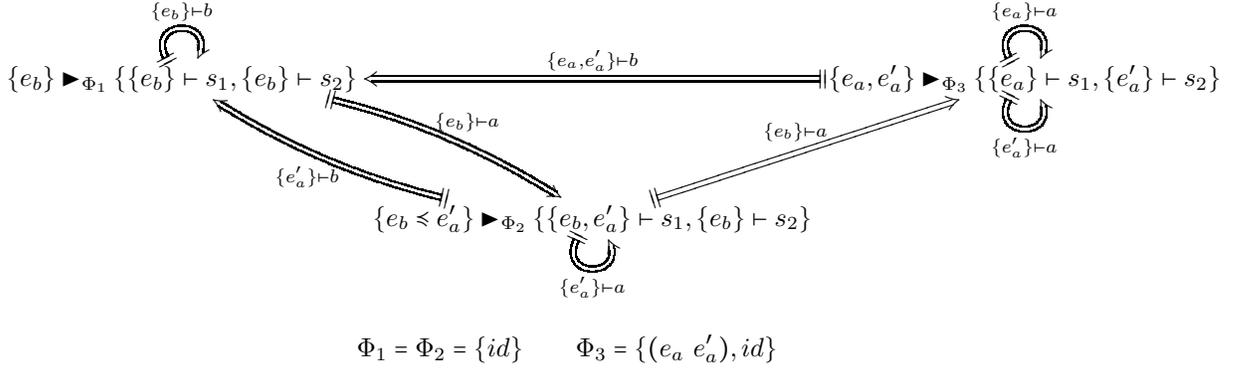

\end{example}

\section{Causal case graphs and behavior structures}
\label{sec:bs}

In the pioneering work \cite{TR88} of Trakhtenbrot and Rabinovich, \emph{behavior structures} have been introduced as causal models for Petri nets. In this section we compare them with our causal models. We recall a slightly simplified definition.
\begin{definition}[Behavior structure]
\label{def:beh-str}
Let $Act$ be a set of action labels. A behavior structure (BS in short) is a triple $B = (M,P,\phi)$, where:
\begin{itemize}
    \item $M$ is an automaton such that:
    \begin{itemize}
        \item transitions have the form $n \bstrans{B}{a} m$, with $a \in Act$;
        \item all states are reachable from the initial one $r$;
        \item there are no oriented cycles, i.e., sequences of transitions where the first and last state coincide;
        \item there are no parallel edges, i.e., $n \bstrans{B}{a} m$ and $n \bstrans{B}{b} m$ implies $a = b$.
    \end{itemize}
    \item $P$ is a family $P_n$ of $Act$-labeled posets of events, one for each state $n$ of $M$ (for the root state $r$ we must have $P_r = \emptyset$);
    \item $\phi$ is a family of labeled posets morphisms: for each pair of states $n$ and $m$ such that $n \bstrans{B}{a} m$
    \begin{itemize}
        \item $\phi_{n,m}$ is an isomorphic embedding of $P_n$ as a prefix of $P_m$;
        \item $|P_m| \setminus |\phi_{n,m}(P_n)| = \{ e_a \}$, for some event $e$;
    \end{itemize}
    \end{itemize}
\end{definition}
In a BS, each state $n$ has a poset $P_n$ over labeled events, describing causal dependencies among occurrences of actions that led to $n$. For each transition $n \bstrans{B}{a} m$ we have a map $\phi_{n,m}$ telling the correspondence between $P_n$ and $P_m$: $P_n$ is required to be isomorphic to a prefix of $P_m$ because it should specify causal dependencies for all the previous actions. The only additional event in $P_m$ represents an occurrence of the most recent action $a$.

The associated notion of behavioral equivalence is called \emph{BS-bisimilarity}. In \cite{TR88}, this equivalence  compares two different behavior structures. Here states belong to the same behavior structure.
\begin{definition}[BS-bisimulation]
\label{def:bsbisim}
Given a behavior structure $B$, a \emph{BS-bisimulation} on $B$ is a relation $R$ on triples such that, whenever $(n_1,\sigma,n_2) \in R$:
\begin{itemize}
    \item $\sigma$ is an isomorphism between $P_{n_1}$ and $P_{n_2}$;
    \item if $n_1 \bstrans{B}{a} m_1$ then there exist $m_2,\sigma'$ such that $n_2 \bstrans{B}{a} m_2$ with $(m_1,\sigma',m_2) \in R$ and the following diagram commutes
    \[
        \xymatrix{
            P_{n_1} \ar[r]^\sigma \ar[d]_{\phi_{n_1,m_1}} & P_{n_2} \ar[d]^{\phi_{n_2,m_2}} \\
            P_{m_1} \ar[r]_{\sigma'} & P_{m_2}
        }
    \]
    (and viceversa)
\end{itemize}
The greatest such relation, denoted $\bsbisim$, is called \emph{BS-bisimilarity}.
\end{definition}
Notice that states are related by BS-bisimulations up to an isomorphism of their posets. This is because the actual identity of events should not matter when comparing states. Only the causal dependencies between occurrences of actions are relevant. BS-bisimilarity has been called \emph{history preserving bisimilarity} \cite{FroschleH99} in later work.

\subsection{Relationship with causal case graphs}
\newcommand{\mc}{M^{\mathtt{C}}}
\newcommand{\phic}{\phi^{\mathtt{C}}}
\newcommand{\psc}{P^{\mathtt{C}}}

When used to represent the behavior of Petri nets, states of behavior structures are states of deterministic, non-sequential processes equipped with information about the past history of events. They can equivalently be seen as tokens equipped with causal information (see, e.g., \cite{MontanariP97}). Therefore, we will consider behavior structures over causal markings. This will enable a more direct comparison with our causal case graphs.

We characterize a sub-LTS of \clts{} that is equivalent to a BS.
\begin{definition}[Reachable \clts]
The \emph{reachable} \clts{} (\rclts) is defined as follows:
\begin{itemize}
    \item it has an \emph{initial} P-marking $\emptyset \rhd \emptyset \caus m_0$, where $m_0$ is an initial marking for $N$;
    \item transitions are only those reachable from $\emptyset \rhd \emptyset \caus m_0$.
\end{itemize}
\end{definition}
\rclts{} enjoys some properties that allow us to define a BS on top of it.
\begin{lemma}
\label{lem:rclts-prop}
\hfill
\begin{enumerate}[$(i)$]
    \item Each state $O_c \rhd c$ of \rclts{} has a unique possible poset, i.e., for any other state $O \rhd c$ we have $O=O_c$; moreover, we have $|O_c| = \causes(c)$.
    \item \rclts{} does not have parallel transitions and directed cycles.
\end{enumerate}
\end{lemma}

\begin{proposition}
\label{prop:bsc}
The triple $\bsc = (\mc,\phic,\psc)$ is a behavior structure, where
\begin{itemize}
    \item $\mc$ is the smallest automaton generated from \rclts{} via the following rule
\[
    \frac
    {
        O_{c} \rhd c \ctrans{K}{e}{a} O_{c'} \rhd c'
    }
    {
        c \bstrans{\bsc}{a} c'
    }
\]
\item $\psc = \{ O_c \mid \text{$O_c \rhd c$ is a state of \rclts} \}$;
\item $\phic = \{ \phic_{c,c'} : O_{c} \subto O_{c'} \mid  O_{c} \rhd c \ctrans{K}{e}{a} O_{c'} \rhd c' \}$.
\end{itemize}
\end{proposition}
We have the following relation between $\cbisim$ and BS bisimilarity. 
\begin{theorem}
\label{thm:bs-c-bisim}
Let $c_1,c_2$ be states of \bsc. Then
\begin{enumerate}[$(i)$]
    \item If $O \rhd c_1 \cbisim O \rhd c_2$ and there is an isomorphism $\sigma \colon O_{c_1} \to O_{c_2}$ then $c_1 \bsbisim^{\sigma} c_2$;
    \item $c_1 \bsbisim^\sigma c_2$ implies $O_{c_2} \rhd c_1\sigma \cbisim O_{c_2} \rhd c_2$. 
\end{enumerate}
\end{theorem}
Statement (i) says that two states $c_1$ and $c_2$ in $\bsc$ with isomorphic posets are $\bsbisim$-bisimilar whenever any two P-markings over $c_1$ and $c_2$ are $\cbisim$-bisimilar. Statement (ii) is somewhat dual: if $c_1$ and $c_2$ are $\bsbisim$-bisimilar under an isomorphism $\sigma$, then we can use $\sigma$ to turn them into $\cbisim$-bisimilar P-markings.
\begin{remark}
The behavior structure we have introduced has some common aspects with \iclts{}: for both, posets in states have local meanings; in fact, bisimilarities require explicit mappings between posets of simulating states. However, \iclts{} can discard event names along transitions and go back to an already visited state, whereas this is explicitly forbidden for BSs.
\end{remark}

\section{Background on category theory}
\label{sec:back}

We assume that the reader is familiar with elementary category theory. In this section we recall some notions that will be needed in the following.

\subsection{Functor categories}
\label{ssec:funcs}

\begin{definition}[Functor category]
Let $\catC$ and $\catD$ be two categories. The \emph{functor category} $\catD^\catC$ has functors $\catC \rightarrow \catD$ as objects and natural transformations between them as morphisms.
\end{definition}
Functors from any category $\catC$ to $\catSet$ are called (covariant) \emph{presheaves}. Hereafter we assume that the domain category $\catC$ for presheaves is \emph{small}, i.e., its collection of objects is actually a set. A presheaf $P$ can be intuitively seen as a family of sets indexed over the objects of $\catC$ plus, for each $\sigma \colon c \to c'$ in $\catC$, an action of $\sigma$ on $Pc$, which we write 
\[
	p[\sigma]_P = P\sigma(p) \qquad (p \in Pc) \enspace ,
\]
omitting the subscript $P$ in $[\sigma]_P$ when clear from the context. This notation intentionally resembles the application of a renaming $\sigma$ to a process $p$, namely $p\sigma$: it will, in fact, have this meaning in later sections.
The set $\el P$ of \emph{elements} of a presheaf $P$ is
\[
	\el P := \sum_{c \in |\catC|} Pc
\]
where the sum symbol denotes the coproduct in $\catSet$, and we denote by $c \rhd p$ a pair belonging to $\el P$.
Presheaf categories have the following nice property.
\begin{property}
For any $\catC$, $\psh{\catC}$ has all limits and colimits, both computed pointwise.
\end{property}

\subsection{Coalgebras}
\label{ssec:coalg}

The \emph{behavior} of systems can be modeled in a categorical setting through \emph{coalgebras} \cite{Rutten00,Adamek05}. Given a \emph{behavioral endofunctor} $B \colon  \catC \to \catC$, describing the ``shape'' of a class of systems, we have a corresponding category of coalgebras.
\begin{definition}[$\catCoalg{B}$]
The category $\catCoalg{B}$ is defined as follows: objects are \emph{$B$-coalgebras}, i.e., pairs $(X,h)$ of an object $X \in |\catC|$, called \emph{carrier}, and a morphism $h \colon X \to BX$, called \emph{structure map}; \emph{$B$-coalgebra homomorphisms} $f \colon (X,h) \to (Y,g)$ are morphisms $f \colon X \to Y$ in $\catC$ making the following diagram commute
\[
	\xymatrix{
		X \ar[r]^{h} \ar[d]_f & BX \ar[d]^{Bf} \\
		Y \ar[r]_{g} & BY
	}
\] 
\end{definition}
For instance, given a set of labels $L$, consider the functor 
\[
	B_{flts} := \finParts( L \times - ) 
\]
where $\finParts \colon \catSet \to \catSet$ is the \emph{finite powerset functor}, defined on a set $A$ and on a function $h \colon A \to A'$ as follows
\[
	\finParts A := \{ B \subseteq A \mid \text{$B$ finite} \} \qquad
	\finParts h (B):= \{ h(b) \mid b \in B \} 
\]
$B_{flts}$-coalgebras $(X,h)$ are \emph{finitely-branching labeled transition systems}, with labels $L$ and states $X$. The function $h(x)$ returns the set of labeled transitions $x \trans{a} y$ such that $(a,y) \in h(x)$. Homomorphisms of $B_{flts}$-coalgebras are functions between states that preserve and reflect transitions.

Many notions of behavioral equivalence can be defined for coalgebras (see \cite{Staton11}). We adopt the one by Hermida and Jacobs and we simply call it $B$-bisimulation. To introduce it, we need some preliminary notions. A (binary) \emph{relation on $X \in |\catC|$} is a jointly-monic span $X \leftarrow R \rightarrow X$ in $\catC$. An \emph{image} of a morphism $f \colon A \to C$ is a monomorphism $m \colon B \injto C$ through which $f$ factors, such that if $f$ factors through any other mono $B' \injto C$, then $B$ is a subobject of $B'$. The factoring morphism $A \to B$ is called \emph{cover}. In $\catSet$ all these notions become the usual ones: a relation $R$ is a binary relation on $X$ and the span is made of left/right projections; the image of $f$ is $f(A) \subto C$, and its cover is $f$ with restricted codomain $f(A)$. Given a relation $R$ on $X$, the \emph{relation lifting} $\overline{B} R$ is the image of the morphism $BR \to B(X \times X) \to BX \times BX$, taking $R$ to a relation on $BX$.
\begin{definition}[$B$-bisimulation]
Given a $B$-coalgebra $(X,h)$, a \emph{$B$-bisimulation} on it is a relation $R$ on $X$ such that there is $r$ making the following diagram commute
\[
	\xymatrix{
		X \ar[d]_h & \ar[l] R \ar[r] \ar@{..>}[d]_r & X \ar[d]^h \\
		BX & \ar[l] \overline{B}R \ar[r] & BX \\
	}
\]
The greatest such relation is called $B$-bisimilarity.
\label{def:coalg-bisim}
\end{definition}    
A $B_{flts}$-bisimulation $R$ on a $B_{flts}$-coalgebra is an ordinary bisimulation on the corresponding transition system. In fact, $\overline{B}R$ is the set of pairs $(X_1,X_2) \in BX \times BX$ such that
$(l,x') \in X_1$ only if there is some $(l,(x',y')) \in BR$, but then we also have $(l,y') \in X_2$ and $(x',y') \in R$ (the symmetric statement holds if $(l,x') \in X_2$). Clearly $r$ exists if and only if $R$ is a bisimulation, and is given by $(x,y) \in R \mapsto (h(x),h(y))$.

An important property of categories of coalgebras is the existence of the terminal object; the unique morphism from each coalgebra to it assigns to each state its abstract semantics. The ideal situation is when the induced equivalence, relating all the states with the same abstract semantics, agrees with $B$-bisimilarity. A sufficient condition for this property is when \emph{$B$ covers pullbacks}.
\begin{property}[$B$ covers pullbacks]
Consider a cospan $X_1 \rightarrow X_3 \leftarrow X_2$, and the morphism $m$ from the image of the pullback (the left square below) to the pullback of the image
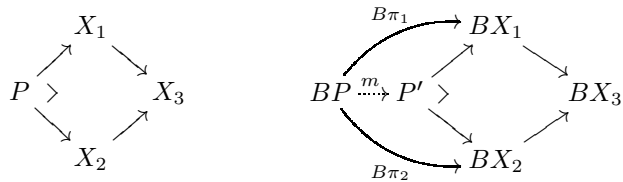
\[
    \xymatrix@=2.5ex{
        & X_1 \ar[dr] & \\
        P \pullbackcorner[r] \ar[ur] \ar[dr] & & X_3 \\
        & X_2 \ar[ur]
    }
    \qquad\qquad
    \xymatrix@=2.5ex{
        && BX_1 \ar[dr] & \\
        BP \ar@{..>}[r]^m \ar@/^1pc/[urr]^{B\pi_1} \ar@/_1pc/[drr]_{B\pi_2} &P' \pullbackcorner[r] \ar[ur] \ar[dr] & & BX_3 \\
        && BX_2 \ar[ur]
    } 
\]
Then $B$ \emph{covers pullbacks} if $m$ is always a cover.
\end{property}
For the best-known Aczel-Mendler bisimulations, defined as spans of coalgebras, the condition on $B$ that guarantees the agreement of behavioral equivalences is more demanding: $B$ should preserve weak pullbacks. The finite powerset functor on $\catSet$ preserves weak pullbacks, but other finite powerset functors do not, for instance the one on presheaves that we will use, which instead covers pullbacks. This motivates our preference of Hermida-Jacobs bisimulations over Aczel-Mendler ones (another important reason for this will be explained in \autoref{sec:coalg}).

A sufficient condition for the existence of the final coalgebra is that $B$ is an \emph{accessible} functor on a \emph{locally finitely presentable} category (see \cite{AdamekR94,Worrell99,Adamek05} for details). A category $\catC$ is \emph{filtered} if each finite diagram is the base of a cocone in $\catC$; filtered categories generalize the notion of directed preorders, that are sets such that every finite subset has an upper bound. For any category $\catD$, a \emph{filtered colimit} in $\catD$ is the colimit of a diagram of shape $\catC$, i.e., a functor $\catC \to \catD$, such that $\catC$ is a filtered category.
\begin{definition}[Locally finitely presentable category] An object $c$ of a category $\catC$ is \emph{finitely presentable} if the functor $\homFun{\catC}(c,-) \colon \catC \to \catSet$ preserves filtered colimits. A category $\catC$ is locally finitely presentable if it has all colimits and there is a set of finitely presentable objects $X \subseteq | \catC |$ such that every object is a filtered colimit of objects from $X$. 
\end{definition}
For instance, locally finitely presentable objects in $\catSet$ are precisely finite sets. $\catSet$ is locally finitely presentable: every set is the filtered colimit, namely the union, of its finite subsets and the whole $\catSet$ is generated by the set containing one finite set of cardinality $n$ for all $n \in \mathbb{N}$.

For functor categories we have the following.
\begin{proposition}
For each locally finitely presentable category $\catC$ and small category $\catD$, the functor category $\catC^\catD$ is locally finitely presentable.
\end{proposition}
In particular, since $\catSet$ is locally finitely presentable, we have that the presheaf category $\catSet^\catD$ is locally finitely presentable as well.
\begin{definition}[Accessible functor] Let $\catC$ and $\catD$ be locally finitely presentable categories. A functor $F \colon \catC \to \catD$ is \emph{accessible} if it preserves filtered colimits. 
\end{definition}
Here are some useful properties of accessible functors: their products, coproducts and composition is accessible as well; adjoint functors between locally finitely presentable categories are accessible. Moreover, it is a well-known fact that the finite powerset functor $\finParts$ introduced in \autoref{ssec:coalg} is accessible.

\subsection{Coalgebras over presheaves}

Coalgebras for functors $B \colon \pshC \to \pshC$ are pairs $(P,\rho)$ of a presheaf $P \colon \catC \to \catSet$ and a natural transformation $\rho \colon P \to BP$. The naturality of $\rho$ imposes a constraint on behavior
\[
	\xymatrix{
		c \ar[d]_{f} & 
		p \in Pc \ar@{|->}[d]_{[f]_P} \ar@{|->}[r]^{\rho_c} & 
		beh(p) \ar@{|->}[d]^{[f]_{BP}} \\
		c' & p[f]_P \in P(c') \ar@{|->}[r]_{\rho_{c'}} & 
		beh(p)[\sigma]_{BP}
	}
\]
Intuitively, this diagram means that, if we take a state, apply a function to it and then compute its behavior, we should get the same thing as first computing the behavior and then applying the function to it. In other words, behavior must be \emph{preserved} and \emph{reflected} by the index category morphisms. 

$B$-bisimulations have a similar structure. A $B$-bisimulation $R$ is a presheaf in $\pshC$ and all the legs of the bisimulation diagram in \autoref{def:coalg-bisim} are natural transformations. In particular, the naturality of projections implies that, given $(p,q) \in Rc$ and $f \colon c \to c'$ in $\catC$, $(p[f],q[f]) \in R(c')$, i.e., $B$-bisimulations are \emph{closed under the index category morphisms}.

\section{Coalgebraic semantics}
\label{sec:coalg}

In this section we construct a coalgebraic causal semantics for Petri Nets. We first show that the notions of \autoref{ssec:aclts} have a categorical interpretation. Then we translate \aclts{} into a coalgebra.

We introduce two categories of $Act$-labeled posets. Recall that, given a category $\catC$, a \emph{skeletal category} is a full subcategory of $\catC$ such that each object is isomorphic to one of $\catC$ and two distinct objects cannot be isomorphic.
\begin{definition}[Category $\cat{O}$ and $\catO$]
Let $\cat{O}$ be the skeletal category of the category of $Act$-labeled posets and their morphisms. The category $\catO$ is the subcategory of $\cat{O}$ whose morphisms are order-embeddings.
\end{definition}
Taking a skeletal category amounts to choosing one canonical representative of each isomorphism class of posets, i.e., using the terminology of \autoref{ssec:aclts}, the objects of $\cat{O}$ and $\catO$ are abstract posets. The difference between $\cat{O}$ and $\catO$ is similar to that between $\mathbf{F}$, the category of finite ordinals and all functions, and its subcategory $\mathbf{I}$, including only injective functions (indeed $\catO$ only includes injective morphisms). Presheaves over these categories are used in \cite{FioreT01} to give a coalgebraic semantics for the $\pi$-calculus.

\begin{remark}
In \cite{ACTA} we have introduced the category $\catPO$ of finite posets up to isomorphisms and its subcategory $\catPOm$ with only order-embeddings. The category $\cat{O}$ can be understood as a comma category $U \downarrow Act$, where $U \colon \catPO \to \catSet$ takes a poset to its underlying set and $Act$ is the constant functor mapping every set to $Act$. Similarly for $\catO$, whenever $U\colon \catPOm \to \catSet$.
\end{remark}

\begin{proposition}
\label{}
The category $\catO$ is small and has pullbacks.
\label{prop:o-small-pull}
\end{proposition}
The category $\catO$ lacks colimits, but the ones we are interested in can be computed in $\cat{O}$. We will be more precise when presenting such colimits.

We introduce some notation for particular objects and morphisms of $\cat{O}$. We denote by $\dpos{k}{l}$ the discrete poset with $k$ elements and labeling function $l$; if $k=1$ then we simply write $\dpos{1}{a}$ to assign label $a$ to the only event. We write $\dposT{k}{l}{a}$ for the poset $\dpos{k}{l}$ plus a top element with label $a$. Two maps will be useful: 
\[
    \xymatrix
    {
        \dpos{k}{l} \ar[r]^{b(\dposT{k}{l}{a})} & \dposT{k}{l}{a} & \ar[l]_{\top(\dposT{k}{l}{a})} \dpos{1}{a}
    }
\]
the left map picks the bottom elements in $\dposT{k}{l}{a}$, and the right one picks the top element.

In $\cat{O}$ we can use a pushout to compute $\delta(O,K,a)$, the associated maps $old(O,K,a)$ and $new(O,K,a)$, and the extension $\ext{\sigma}$ of a morphism $\sigma \colon O \to O'$, all defined in \autoref{ssec:aclts}. Given $O \in |\cat{O}|$, let $K \colon \dpos{k}{l} \hookrightarrow O$ be the subobject in $\catO$ picking $K$ within $O$. Then we have
\begin{equation}
	\begin{gathered}
	\xymatrix{
		\dpos{k}{l} 
		\ar[r]^{K}  
		\ar[d]_{b(\dposT{k}{l}{a})} 
		& 
		O 
		\ar[d]^{old(O,K,a)} 
		\ar[r]^{\sigma} 
		& 
		O' 
		\ar[dd]^{old(O',\sigma(K),a)}  
		\\
		\dposT{k}{l}{a} 
		\ar[d]_{id}
		\ar[r]_{K^a} 
		& 
		\pushoutcorner
		\delta(O,K,a) 
		\ar@{-->}[dr]^{\ext{\sigma}} 
		\\
		\dposT{k}{l}{a} 
		\ar[rr]_{(\sigma(K))^a}
		&& 
		\pushoutcorner\delta(O,\sigma(K),a)
	}
	\end{gathered}
	\qqquad
	new(O,K,a) = K^a \circ \top(\dposT{k}{l}{a})
	\label{diag:delta}
\end{equation}
Explicitly, $\delta(O,K,a)$ is constructed as follows: the disjoint union of $O$ and $\dposT{k}{l}{a}$ is made, and then the bottom elements of $\dposT{k}{l}{a}$ and the causes $K$ are identified, resulting in $O$ plus a fresh $a$-labeled top event for $K$; the transitive closure of this relation gives $\delta(O,K,a)$. Notice that, since $K$ reflects order, causes of the fresh event must be incomparable, i.e., they are maximal events in $O$. This agrees with the definition of $K$ in \autoref{def:clts}. The map $\ext{\sigma} \colon \delta(O,K,a) \to \delta(O',\sigma(K),a)$ is induced by the universal property of pushouts: we compute $\delta(O',\sigma(K),a)$ via the pushout of
\[
    \xymatrix@1{\dpos{k}{l} &\ar[l]_{b(\dposT{k}{l}{a})} \dposT{k}{l}{a} \ar[r]^{\sigma \circ K} & O'} 
\]
that is the outer pushout in \eqref{diag:delta}, and then we define $\ext{\sigma}$ as the mediating morphism between the inner and the outer pushout. It can be easily verified that $\ext{\sigma}$ indeed acts as described in \autoref{ssec:aclts}. All these constructions has been given in $\cat{O}$ but we have the following property.

\begin{lemma}
The diagram \eqref{diag:delta} also exists in $\catO$.
\label{lem:o-delta-diag}
\end{lemma}

Now we want to turn the computation of $\delta(O,K,a)$ into a functorial operation on $\catO$. This operation can only have $O$ as parameter. The dependency from $a$ and $K$ is removed by adding a new event for \emph{each} set of independent causes and \emph{each} action. Formally, consider all $K_1 \colon \dpos{k_1}{l_1} \subto O, \dots , K_m \colon \dpos{k_m}{l_m} \subto O$. Suppose $Act = \{a_1,\dots,a_n\}$. Then we can compute $\deltaFun(O)$ via the colimit shown in \autoref{fig:colimit-delta}.
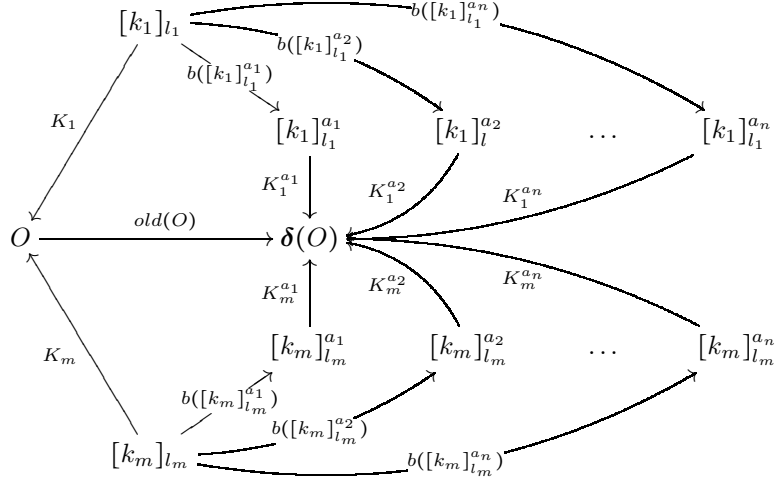
\begin{figure}[t]
\[
    \xymatrix{
        & \dpos{k_1}{l_1} 
        \ar[dr]|{b(\dposT{k_1}{l_1}{a_1})} 
        \ar@/^1pc/[drr]|{b(\dposT{k_1}{l_1}{a_2})} 
        \ar@/^2pc/[drrrr]|{b(\dposT{k_1}{l_1}{a_n})}         
        \ar[ddl]_{K_1}
        \\
        &&
        \dposT{k_1}{l_1}{a_1}         
        \ar[d]_{K_1^{a_1}}
        & 
        \dposT{k_1}{l}{a_2} 
        \ar@/^1pc/[dl]_{K_1^{a_2}}
        & 
        \dots 
        & 
        \dposT{k_1}{l_1}{a_n} 
        \ar@/^1pc/[dlll]_{K_1^{a_n}}
        \\
        O \ar[rr]^{old(O)} && \deltaFun(O)
        \\
        &&
        \dposT{k_m}{l_m}{a_1}    
        \ar[u]^{K_m^{a_1}}     
        & 
        \dposT{k_m}{l_m}{a_2} 
        \ar@/_1pc/[ul]^{K_m^{a_2}}     
        & 
        \dots 
        & 
        \dposT{k_m}{l_m}{a_n} 
        \ar@/_1pc/[ulll]^{K_m^{a_n}}     
        \\
        &
        \dpos{k_m}{l_m}
        \ar[uul]^{K_m}
        \ar[ur]|{b(\dposT{k_m}{l_m}{a_1})} 
        \ar@/_1pc/[urr]|{b(\dposT{k_m}{l_m}{a_2})} 
        \ar@/_2pc/[urrrr]|{b(\dposT{k_m}{l_m}{a_n})}                
    }
\]
\caption{Colimit computing $\deltaFun(O)$.}
\label{fig:colimit-delta}
\end{figure}
It is the colimit of $m$ cospans with vertex $\dpos{k_i}{l_i}$. Each cospan is similar to the cospan in (\ref{diag:delta}), but its legs include all morphisms $K_i^{a} \colon \dposT{k_i}{l_i}{a} \to \deltaFun(O)$, for all $a \in Act$, instead of a single morphism for a given $a$. This means that, for each set of causes $K_i$, in $\deltaFun(O)$ we have fresh events labeled by all possible actions.

Notice that $\deltaFun(O)$ and $old(O)$ do not depend on $K$ and $a$. We can recover $new$ maps as follows
\[
    new(O,K_i,a) = K^a_i \circ \top(\dposT{k}{l_i}{a}) \colon \dpos{1}{a} \to \deltaFun(O)
\]
Given a morphism $\sigma \colon O \to O'$, we denote $\deltaFun(\sigma) \colon \deltaFun(O) \to \deltaFun(O')$ the corresponding morphism induced by the universal property of the above colimit. Since the colimit in \autoref{fig:colimit-delta} is formed by many diagrams like the inner pushout in (\ref{diag:delta}), by the universal property of pushouts there are unique maps 
\[
\epsilon(O,K_i,a) \colon \delta(O,K_i,a) \to \deltaFun(O) \enspace .
\]
Then we can relate $\deltaFun(O)$ and each $old(O,K_i,a)$
\[
    old(O) =  \epsilon(O,K_i,a) \circ old(O,K_i,a) \colon O \to \deltaFun(O)
\]
and see how each $\ext{\sigma}$ ``embeds'' into $\deltaFun(\sigma)$, namely
\[
    \xymatrix@C+3ex{
        \delta(O,K_i,a) \ar[d]_{\ext{\sigma}} \ar[r]^-{\epsilon(O,K_i,a)} & \deltaFun(O) \ar[d]^{\deltaFun(\sigma)} \\
        \delta(O',\sigma(K_i),a) \ar[r]_-{\epsilon(O',\sigma(K_i),a)} & \deltaFun(O')
    }
\]
The intuition is that $\deltaFun(\sigma)$ acts as $\sigma$ on old events (as all $\ext{\sigma}$ do) and as the specific $\ext{\sigma}$ on new ones. Since each $\ext{\sigma}$ is an order-embedding (\autoref{lem:o-delta-diag}), also $\deltaFun(\sigma)$ is, so $\deltaFun(\sigma)$ is a morphism of $\catO$. This means that $\delta$ defines a proper \emph{allocation endofunctor} on $\catO$.
\begin{example}
Suppose $Act = \{c,d\}$ and let $O$ be the discrete abstract poset $\{e_a,e'_b\}$. Then $\deltaFun(O)$ contains $new(O,\emptyset,c)$, $new(O,\emptyset,d)$, and the following pairs (we omit reflexive ones):
\begin{alignat*}{3}
    e_a &\lss new(O,\{e_a\},c) 
    &&& \qqquad e_a &\lss new(O,\{e_a\},d) 
    \\
    e_a &\lss new(O,\{e_a,e'_b\},c)
    &&& \qqquad e_a &\lss new(O,\{e_a,e'_b\},d)
    \\[2ex]
    e'_b &\lss new(O,\{e'_b\},c) 
    &&& \qqquad e'_b &\lss new(O,\{e'_b\},d) 
    \\
    e'_b &\lss new(O,\{e_a,e'_b\},c)
    &&& \qqquad e'_b &\lss new(O,\{e_a,e'_b\},d)
\end{alignat*}
\end{example}
\begin{remark}
Our definition of $\deltaFun$ may not seem the best one, as it generates a new event for each possible set of causes and each label, whereas a transition only generates one of these events. However, having a functor on $\catO$ allows us to lift it to presheaves in a way that ensures the existence of both left and right adjoint (giving Kan extensions along $\deltaFun$) for the lifted functor, and then preservation of both limits and colimits, which is essential for coalgebras employing such functor. Generation of unused events is not really an issue: as we will see later, it is always possible to recover the support of a P-marking, i.e., the poset formed by events actually appearing in it.
\end{remark}
Now we look at the category $\pshO$ of presheaves on labeled posets. Since $\catO$ is small it follows that $\pshO$ is locally finitely presentable and has all limits and colimits, in particular products and coproducts. The following functors are relevant for us.

\paragraph{Presheaf of event names} $\enames \colon \catO \to \catSet$ maps $O$ to the set $|O|$. Formally
\[
	\enames = \sum_{a \in Act} \Hom_{\catO}(\dpos{1}{a},-)
\]
where $e_a \in |O|$ is represented as a morphism $\dpos{1}{a} \to O$. The action of $\enames$ on a morphism $\sigma\colon O \to O'$ gives the function $\lambda e_a \in \enames(O).\sigma \circ e_a$, which renames the event $e_a$ according to $\sigma$.

\paragraph{Finite powerset} $\FinParts \colon \pshO \to \pshO$, defined as $\mathcal{P}_f \circ (-)$, where $\mathcal{P}_f$ is the finite powerset on $\catSet$.
\paragraph{Event allocation operator} $\Delta \colon \pshO \to \pshO$, given by $(-) \circ \delta$. Explicitly, for $P \colon \catO \to \catSet$ and $O \in |\catO|$, $\Delta P (O) = P( \deltaFun(O) )$. Intuitively, it generates causal markings with additional fresh events.
\paragraph{Presheaf of labels} $\lpsh \colon \catO \to \catSet$ given by
\[
		 \lpsh(O) = Act \times \FinParts \enames(O)
\]
For each $O \in |\catO|$, this functor gives pairs $(a,K)$ of an action $a$ and a finite set of causes $K$, selected among events in $O$. 

We use these operators to define our behavioral endofunctor.
\begin{definition}[Behavioral functor]
\label{def:beh-fun}
The behavioral functor $B \colon \pshO \to \pshO$ is
\[
	B P = \FinParts (  \lpsh \times \Delta P ) \enspace .
\]
\end{definition}
To understand this definition, consider a $B$-coalgebra $(P,\rho)$. Given $O \in |\catO|$ and $p \in P(O)$, $\rho_O(p)$ is a finite set of triples $(a,K,p')$, meaning that $p'$ is the continuation of $p$ after observing $K \caus a$. The continuation always belongs to $\Delta P (O)$, because every transition allocates a new event.

The category $\catCoalg{B}$ is well-behaved: it has a final $B$-coalgebra, and the behavioral equivalence it induces coincides with $B$-bisimilarity. This is thanks to the following properties.
\begin{proposition}
\label{prop:B-acc}
	$B$ is accessible and covers pullbacks.
\end{proposition}
$B$-coalgebras can be regarded as particular LTSs whose states are elements of presheaves, i.e., pairs $O \rhd p$. 
\begin{definition}[\oilts]
\label{def:po-ilts}
An \emph{$\catO$-indexed labeled transition system} (\oilts) is a pair $(P,\Trans{})$ of a presheaf $P \colon \catO \to \catSet$ and a finitely-branching transition relation $\Trans{\quad} \subseteq  \el P \times \el \lpsh \times \el P$ of the form:
\[
	O \rhd p \Trans{K \caus a} \deltaFun(O) \rhd p' \qquad (a,K) \in \lpsh (O) 
\]
such that, for each morphism $\sigma\colon O \to O'$ in $\catO$:
\begin{enumerate}[$(i)$]
	\item if $O \rhd p \Trans{l} \deltaFun(O) \rhd p'$ then $O' \rhd p[\sigma] \Trans{l[\sigma]} \deltaFun(O') \rhd p'[\deltaFun(\sigma) ]$ (transitions are \emph{preserved} by $\sigma$);
	\item if $O' \rhd p[\sigma] \Trans{l} \deltaFun(O') \rhd p'$ then there are $l'$ and $\deltaFun(O) \rhd p''$ such that $l'[\sigma] = l$, $p''[\deltaFun(\sigma)] = p'$ and $O \rhd p \Trans{l'} \deltaFun(O) \rhd p''$ (transitions are \emph{reflected} by $\sigma$);
\end{enumerate}
\end{definition}
Now, notice that labels and continuations of \oilts s agree with those generated by $B$, and (i) and (ii) say that the transition relation behaves like a natural transformation. Therefore we have the following correspondence.
\begin{proposition}
\label{prop:coalg-poilts}
\oilts s are in bijection with $B$-coalgebras.
\end{proposition}
The natural notion of bisimulation for these transition systems is \emph{$\catO$-indexed bisimulation}.
\begin{definition}[$\catO$-indexed bisimulation]
\label{def:oibisim}
An \emph{$\catO$-indexed bisimulation} on an \oilts{} $(P,\Trans{\quad})$ is an indexed family of relations $\{ R_O \subseteq P(O) \times P(O) \}_{O \in |\catO|}$ such that, for all $(p,q) \in R_O$:
\begin{enumerate}[$(i)$]
	\item if $O \rhd p \Trans{K \caus a} \deltaFun(O) \rhd p'$ then there is $\deltaFun(O) \rhd q'$ such that $O \rhd q \Trans{K \caus a} \deltaFun(O) \rhd q'$ and $(p',q') \in R_{\deltaFun(O)}$;
	\label{oibisim-sim}
	\item for all $\sigma \colon O \to O'$, $(p,q) \in R_O$ if and only if $(p[\sigma]_P,q[\sigma]_P) \in R_{O'}$.
	\label{oibisim-clos}
\end{enumerate}
\end{definition}
This definition closely resembles that of \acb-bisimulations (\autoref{def:acbisim}). We have an additional condition \ref{oibisim-clos}, requiring closure under morphisms of $\catO$. This is not satisfied by all \acb-bisimulations, but it holds for the greatest one (\autoref{thm:acbisim-clos}).
We have the following correspondence.
\begin{proposition}
Let $(P,\rho)$ be a $B$-coalgebra. Then $B$-bisimulations on $(P,\rho)$ are in bijection with $\catO$-indexed bisimulations on the induced \oilts.
\label{prop:bbsim-obisim-eq}
\end{proposition}
Notice that, unlike Aczel-Mendel bisimulations, a $B$-bisimulation (namely, a Hermida-Jacobs one) needs not be the carrier of a $B$-coalgebra in order to be a bisimulation. This strong requirement is the reason why some $\catO$-indexed bisimulations cannot be turned into Aczel-Mendler ones (see \cite[3.3, Anomaly]{Staton07}).

We now show that \aclts{} can be represented as an \oilts. We form a presheaf from P-markings as follows.
\begin{definition}[Presheaf of P-markings]
The \emph{presheaf of P-markings} $\pshMark \colon \catO \to \catSet$ is given by
\[
	\pshMark (O) = \{ c \mid \text{$O \rhd c$ is an abstract P-marking} \}
	\qquad
	\pshMark(\sigma \colon O \to O') = \lambda ( O \rhd c).O' \rhd \dclos{ (c \sigma) }{O'}
\]
\end{definition}
The action of $\pshMark$ on morphisms needs to apply the closure operator, after renaming the causal marking: this guarantees that the result is a proper P-marking. The functor $\pshMark$ has the following useful property.
\begin{lemma}
\label{lem:mfun-pull}
$\pshMark$ preserves pullbacks.
\end{lemma}
Intuitively, thanks to this property, if we take $c \in \pshMark(O)$ and all subposets $O'$ of $O$ such that $\pshMark(O')$ contains a ``version'' of $c$ (typically with fewer events) then the set obtained by applying $\pshMark$ to the pullback of these subposets, i.e., to their minimal common subposet, still contains a version of $c$. This will be essential, in the next section, to compute minimal representatives of P-markings.

We are ready to translate \aclts{} to an \oilts{}. 
\begin{definition}[Causal \oiltsC]
\label{def:dd-ilts}
The \emph{Causal} \oilts{} (\oiltsC) $(\pshMark,\Trans{\quad})$ is the smallest one generated by the rule
\[
	\frac
	{
		O \rhd c \actrans{K}{a} \delta(O,K,a) \rhd c'
	}
	{
		O \rhd c \acTrans{K}{a} \deltaFun(O) \rhd c'[\epsilon(O,K,a)]
	}
\]
\end{definition}
This translation does not affect bisimilarities: two states can do the same transitions in \aclts{} if and only if they can do the same transitions also in \oiltsC{}; continuations only differ for an order-embedding, but by \autoref{thm:acbisim-clos} and \autoref{def:oibisim}\ref{oibisim-clos}, the $\catO$-indexed bisimilarity and $\acbisim$ are closed under order-embeddings.

We call \emph{causal} coalgebra the $B$-coalgebra equivalent to $(\pshMark,\Trans{\quad})$. We have the following theorem, which collects the results of this section, instantiated to the causal coalgebra.
\begin{theorem}
$\catO$-indexed bisimulations on $(\pshMark , \Trans{\quad})$ are equivalent to:
\begin{itemize}
	\item  $B$-bisimulations on the causal coalgebra;
	\item \acb-bisimulations closed under order-embeddings.\end{itemize}	
\label{thm:bisim-equiv}
\end{theorem}
In particular, we have that the greatest $\catO$-indexed bisimulation, $B$-bisimilarity on the causal coalgebra and $\acbisim$ are all equivalent, thanks to \autoref{thm:acbisim-clos}. These, by \autoref{prop:B-acc}, are equivalent to the kernel of the unique morphism from the causal coalgebra to the final one.

\section{From coalgebras to HD-automata}
\label{sec:hd}
\newcommand{\symg}{\mathtt{G}}
In order to give a characterization of the causal coalgebra in terms of named sets, we employ the results of \cite{CianciaKM10}. Here authors define a \emph{symmetry group over a category $\catC$} to be a collection of morphisms in $\catC[c,c]$, for any $c \in |\catC|$, which is a group w.r.t.\ composition of morphisms. Then they take families of such groups as their notion of generalized named sets. A first result establishes the equivalence between these families and \emph{coproducts of symmetrized representables}, that are functors of the form
\[
	\sum_{i \in I} \Hom_\catC(c_i,\_)_{/\Phi_i}
\]
where $\Phi_i$ is a symmetry group over $\catC$ with domain $c_i$, and the quotient identifies morphisms that are obtained one from the other by precomposing elements of $\Phi_i$.
These functors, in turn, are shown to be isomorphic to \emph{wide-pullback-preserving} presheaves on $\catC$, a wide pullback being the limit of a diagram with an arbitrary number of morphisms pointing to the same object (pullbacks are a special case, with two such morphisms). The described results are summarized in the following theorem from \cite{CianciaKM10}.
\begin{theorem}
Let $\catC$ be a category that is small, has wide pullbacks, and such that all its morphisms are monic and those in $\catC[c,c]$ are isomorphisms, for every $c \in |\catC|$. Then every wide-pullback-preserving $P \! \in \!|\catSet^\catC|$ is equivalent to a coproduct of symmetrized representables.
\label{thm:hd-crit}
\end{theorem}
Our category $\catO$ satisfies the hypothesis of this theorem: it is small and has wide pullbacks due to the existence of pullbacks. In fact, the diagram of a wide pullback in $\catO$ is formed by a finite number of morphisms, because a finite poset always has a finite number of ingoing poset-reflecting monomorphisms, so its limit can be computed via binary pullbacks. Moreover, $\catO$ has only monos, as order-embeddings are always monic, and $\catO[O,O]$ clearly has only isomorphisms, for each $O \in |\catO|$. Finally, our presheaf of causal markings $\pshMark$ preserves (wide) pullbacks (\autoref{lem:mfun-pull}), so there exists an equivalent coproduct of symmetrized representables. 

\autoref{thm:hd-crit} indeed describes an equivalence between pullback-preserving presheaves and families, which induces one on coalgebras. We shall now investigate this point. Let $\pshOPP$ be the full subcategory of $\pshO$ formed by pullback-preserving presheaves. We have that our behavioral endofunctor $B$ indeed defines an endofunctor on $\pshOPP$. 
\begin{proposition}
\label{prop:res-pull}
All the endofunctors on $\pshO$ in \autoref{def:beh-fun} can be restricted to endofunctors on $\pshOPP$.
\end{proposition}
Let $B_\diamond \colon \pshOPP \to \pshOPP$ be the restricted behavioral endofunctor. The causal coalgebra is clearly a $B_\diamond$-coalgebra. Restricting to $\pshOPP$ does not affect the final coalgebra: $\catCoalg{B}$ and $\catCoalg{B_\diamond}$ have the same final object and final morphisms from common objects. In fact, the terminal sequence starts from the final presheaf $1$, pointwise defined as the singleton set, which trivially preserves pullbacks, and goes through $B^n(1) = B_\diamond^n(1)$, for any $n$.
\begin{corollary}[of \autoref{thm:hd-crit}]
Let $\widetilde{B}$ be the behavioral endofunctor on families defined by lifting all functors in \autoref{def:beh-fun} along the equivalence. Then the category $\catCoalg{B_\diamond}$ is equivalent to $\catCoalg{\widetilde{B}}$.
\end{corollary}
In particular, the equivalence relates the final $B_\diamond$-coalgebra and the final $\widetilde{B}$-coalgebra, and their final morphisms. Moreover, since kernels are preserved by equivalence, identifications made by the final morphisms are preserved, hence behavioral equivalence is preserved too. 

Now that we have proved that our categorical setting is suitable for HD-automata, we can translate the causal coalgebra to a HD-automaton. We adopt the definition of HD-automaton given in \cite{CianciaM10}: a HD-automaton is a(ny) coalgebra over a named set. We introduce a notion of named set closer to a more traditional one, but indeed equivalent to the families mentioned above. Given a set $S$ of morphism and a morphism $\sigma$ in $\catO$, we write $S \circ \sigma$ for the set $\{ \tau \circ \sigma \mid \tau \in S \}$ (analogously for $\sigma \circ S$).

\begin{definition}[Category $\catSym{\catO}$]
Let $\catSym{\catO}$ be the category defined as follows:
\begin{itemize}
	\item objects $\Phi$ are subsets of $\catO[O,O]$ that are groups w.r.t.\ composition in $\catO$;
	\item morphisms $\Phi_1 \to \Phi_2$ are sets of morphisms $\sigma \circ \Phi_1$ such that $\sigma \colon dom(\Phi_1) \to dom(\Phi_2)$ and $\Phi_2 \circ \sigma \subseteq \sigma \circ \Phi_1$.
\end{itemize}
\end{definition}
\begin{definition}[Category $\catONSet$]
The category $\catONSet$ is defined as follows:
\begin{itemize}
	\item objects are \emph{$\catO$-named sets}, that are pairs $N = (Q_N,\symg_N)$ of a set $Q_N$ and a function $\symg_N \colon Q_N \to |Sym(\catO)|$. The \emph{local poset} of $q \in Q_N$, denoted $\Vert q \Vert$, is $dom(\sigma)$, for any $\sigma \in \symg_N(q)$.
	\item morphisms $f \colon N \to M$ are \emph{$\catO$-named functions}, that are pairs $( h , \Sigma )$ of a function $h \colon Q_N \to Q_M$ and a function $\Sigma$ mapping each $q \in Q_N$ to a morphism $\symg_M(h(q)) \to \symg_N(q)$ in $Sym(\catO)$.
\end{itemize}
\end{definition}
In the rest of this section we give an explicit description of the $\catO$-named set produced from $\pshMark$ by the equivalence. Its elements will be minimal P-markings with symmetries. We will show that the translation from P-markings to minimal ones with symmetries is achieved via categorical constructions. We need the notions of support, seed and orbit.
\begin{definition}[Support and seed]
\label{def:supp-seed}
Given $O \rhd c$, its \emph{support}, denoted $supp(c)$, is the wide-pullback-object of the following morphisms
\[
	\{ \sigma \colon O' \to O \mid \exists O' \rhd c': c' [\sigma]= c \}
\]
Let $\Sigma_{c}$ be the embedding $supp(c) \subto O$ given by the pullback. Then the \emph{seed} of $c$, denoted $seed(c)$, is the unique element of $\pshMark(supp(c))$ such that $seed(c)[\Sigma_{c}] = c$.
\end{definition}
As shown in \cite{CianciaKM10,GadducciMM06}, preservation of pullbacks by $\pshMark$ is essential to ensure existence and uniqueness of seeds. The seed operation achieves the first two properties of minimal P-markings (see \autoref{def:min-pmark}): $seed(c)$ just contains immediate causes for each token and $supp(c)$ contains all and only those causes. This is illustrated by the following example.
\begin{example}
Consider the following P-marking for the running example
\[
   \{ e_a \lss e'_a , e''_a \lss e'''_a \} \rhd \{\{e_a,e'_a\} \caus s_1,\{e''_a,e'''_a\} \caus s_2 \}
\]
which is reachable after firing $t_1$ and $t_2$ twice.
The set of morphisms of \autoref{def:supp-seed} has four elements
\[
    f_1,f_2 \colon \{ e_a \lss e'_a , e''_a\} \to \{ e_a \lss e'_a , e''_a \lss e'''_a \} \qquad
    f_3,f_4 \colon \{e_a,e'_a\} \to \{ e_a \lss e'_a , e''_a \lss e'''_a \}
\] 
\[
    f_1 =
    \begin{cases} 
        e_a \longmapsto e_a \\
        e'_a \longmapsto e'_a \\
        e''_a \longmapsto e'''_a
    \end{cases}
    \qquad
    f_2 =
    \begin{cases} 
        e_a \longmapsto e''_a \\
        e'_a \longmapsto e'''_a \\
        e''_a \longmapsto e'_a
    \end{cases}
    \qquad    
    f_3 =
    \begin{cases}
        e_a \longmapsto e'_a \\
        e'_a \longmapsto e'''_a
    \end{cases}
  	\qquad
  	f_4 =
    \begin{cases}
        e_a \longmapsto e'''_a \\
        e'_a \longmapsto e'_a
    \end{cases}
\]
In fact, we have
\[
\left.
\begin{aligned}
 ( \{ e_a \lss e'_a , e''_a\} \rhd \{ \{e_a,e'_a\} \caus s_1 , \{e''_a\} \caus s_2 \} ) \ [f_1] \\[1ex]
 ( \{ e_a \lss e'_a , e''_a\} \rhd \{ \{e''_a\} \caus s_1 , \{e_a,e'_a\} \caus s_2 \}) \ [f_2] \\[1ex]
 ( \{e_a,e'_a\} \rhd \{ \{ e_a \} \caus s_1 , \{ e'_a \} \caus s_2 \} ) \ [f_3] \\[1ex]
 ( \{e_a,e'_a\} \rhd \{ \{ e'_a \} \caus s_1 , \{ e_a \} \caus s_2 \} ) \ [f_4]
\end{aligned}
\right\} = 
\{ e_a \lss e'_a , e''_a \lss e'''_a \} \rhd \{\{e_a,e'_a\} \caus s_1,\{e''_a,e'''_a\} \caus s_2 \}
\]
Recall that each $[f_i] = \pshMark(f_i)$ is a function that, when applied to a P-marking, replaces events according to $f_i$ and then down-closes the result w.r.t. $\{ e_a \lss e'_a , e''_a \lss e'''_a \}$.
%
%
It is easy to check that the pullback object of all four morphisms is $\{e_a,e'_a\}$, so the corresponding seed is
\[
	\{e_a,e'_a\} \rhd \{ \{ e_a \} \caus s_1 , \{ e'_a \} \caus s_2 \}.
\] 
Notice that two events have been discarded, because they are not immediate causes.
\end{example}
\begin{definition}[Orbit]
The \emph{orbit} of $O \rhd c$ is
\[
	orb(c) = \{ c [\sigma] \mid \sigma \in \catO[O,O] \}
\]
We denote by $[c]^o$ a canonical choice of an element of $orb(c)$.
\end{definition}
The \emph{orbit} of $c$ is the set of causal markings obtained by applying to $c$ all functions induced by poset automorphisms. Automorphisms are isomorphisms, so taking a canonical representative for this orbit achieves the third requirement of minimal P-markings: it amounts to applying the operation $\isorep{O \rhd c}$, i.e., choosing a representative of isomorphism classes for $O \rhd c$.
\begin{definition}
The $\catO$-named set of minimal P-markings is $(M,\symg_M)$, where
\begin{align*}
	M &= \{ supp(c) \RHD [seed(c)]^o \mid O \rhd c \in \el \pshMark \} 
	\\
	\symg_M &= \lambda O \RHD c.\{ \Phi \in |\catSym{\catO}| \mid dom(\Phi) = O \land \forall \sigma \in \Phi : c[\sigma] = c \}
\end{align*}
\end{definition}
The set $M$ is produced from elements of $\pshMark$: for each of these, we compute the seed, and then we only take the canonical representative for the seed's orbit. As explained, the final result is indeed a minimal P-marking $O \RHD c$. This P-marking is associated a symmetry by $\symg_M$, namely $\Phi = \symg_M(O \RHD c)$, so it becomes the P-marking with symmetry $\psym{O}{\Phi}{c}$.

The derivation of an HD-automaton on $(M,\symg_M)$ in $\catCoalg{\widetilde{B}}$ from the causal coalgebra, along the equivalence, is the category-theoretic counterpart of the derivation of \icslts{} from \aclts. The correspondence between \icslts s and coalgebras over named sets is analogous to the $\pi$-calculus case, where we have set-theoretical HD-automata on one side \cite{MontanariP05} and categorical ones, namely coalgebras over named sets, on the other side. The correspondence for the $\pi$-calculus has been worked out in \cite{Ciancia08,CianciaM10}, and the theory introduced therein seems robust enough to accommodate different notions of named sets such as ours. In particular, functors used to define coalgebras over named sets, such as powerset and allocation functors, should be very similar to those defining $\tilde{B}$.

We briefly illustrate the $\tilde{B}$-coalgebra for the running example. The $\catO$-named set $(M,\symg_M)$ is as follows: $M$ includes all P-markings in \autoref{fig:runex-iccg}, and $\symg_M$ returns the symmetry $\{id\}$ for each of them. Transitions are represented as a $\catO$-named function $(h,\Sigma) \colon (M,\symg_M) \to \widetilde{B}(M,\symg_M)$, where $h$ maps each state $\psym{O}{\{id\}}{c}$ to its label and continuation, and $\Sigma(\psym{O}{\{id\}}{c})$ encodes all history maps for outgoing transitions. 

We leave a deeper investigation of the category of $\catO$-named sets and of $\tilde{B}$-coalgebras for future work.

\section{Conclusions}

In this paper we have introduced an approach to derive compact operational models for causality in Petri nets. In order to do this, we have constructed a labeled semantics of Petri nets in terms of causal case graphs, and we have given a procedure to refine them in order to get minimal, possibly finite-state, representations. We have then modeled causal case graphs in a categorical setting, exploiting a nominal representation of causal relations: they are modeled as posets over event names with action labels. Our categorical treatment is simpler and more natural than the set-theoretic one, and employs standard constructs and results for nominal calculi, namely presheaf-based coalgebras and their equivalence with HD-automata. In particular, reducing the state-space and showing that this operation preserves the semantics require some technical effort in the set-theoretic version, whereas the categorical version employs a general construction that automatically performs this reduction in a semantics-preserving way.

Our approach has a practical significance: we show how to synthesize HD-automata from Petri nets, and how to compute minimal realizations for them, in order to detect bisimilar states. As mentioned, minimization of HD-automata is possible in many cases. Even if our approach does not actually provide a way to minimize nets themselves, one can still decide bisimilarity of markings by minimizing their reachable HD-automata and matching the results. 

Finally, our contribution is also methodological: we  provide a further example in which the presheaf/HD-automata framework is successfully applied. We emphasize that this framework is highly parametric and can possibly be useful in many other cases.

\subsection{Related work}

This paper follows a line of research on coalgebraic models of causality, started in \cite{ACTA} by the same authors. The categorical machinery is the same in both papers, namely presheaf-based coalgebras, HD-automata, and the equivalence among them. However, this paper takes a further step towards a general categorical theory of causality. In \cite{ACTA}, in fact, we have provided models for a particular class of causal LTSs, namely Degano-Darondeau ones. In this paper, instead, we treat Petri nets, which are much more general. For instance, unlike Degano-Darondeau LTSs, Petri nets can describe synchronizations of more than two processes. 

In \cite{ACTA} we start from existing set-theoretic models, similar to abstract CGs, whereas the models we introduce here are novel. In both papers we represent causal dependencies as posets over events, but in \cite{ACTA} events are unlabeled and are canonically represented as natural numbers. Here we have labels and we take a more general approach: instead of choosing specific representatives of events, we make abstract CGs parametric in this choice. This requires more technical work and it further validates the categorical approach, where book-keeping details are abstracted away. The categorical environment in this paper is more elaborate than \cite{ACTA}, due to labeling. In particular, event generation is more complex, and is studied in greater detail. Another difference is that here we give conditions under which the model with only immediate causes is finite, whereas in \cite{ACTA} decidability is not treated.

A first version of HD-automata for Petri nets, called \emph{causal automata}, has been introduced in \cite{MontanariP97}. However, their construction is purely set-theoretical and does not include symmetries, so the existence of a minimal model is not guaranteed. This version of HD-automata is similar to what we call immediate causes CG (without symmetries). HD-automata with symmetries were developed for the $\pi$-calculus in \cite{Pistore99,MontanariP05}, and a general categorical treatment was provided in \cite{CianciaM10}. In all these cases nominal structures associated to states are just a sets of (event) names, whereas we have posets, which are more adequate to represent causal dependencies. 

We can cite \cite{CattaniS04} for the introduction of transitions systems for causality whose states are elements of presheaves, intended to model the causal semantics of the $\pi$-calculus as defined in \cite{BorealeS98}. However, the index of a state is a set of names, without any information about events and causal relations. The advantage of our index category is that it allows reducing the state-space in an automatic way, exploiting a standard categorical construction. This cannot be done in the framework of \cite{CattaniS04}.
Finally, an HD-automaton for causality has been described in \cite{CianciaM10}, but it is derived as a direct translation of causal automata and its states do not take into account causal relations.

Other related works are \cite{StatonW10,Winskel99}, where event structures have been characterized as (contravariant) presheaves on posets. While the meaning of presheaves is similar, the context is different: we consider the more concrete realm of coalgebras and nominal automata. A more precise correspondence with such models should be worked out.

\subsection{Future work}

Logics for causality have been recently studied in \cite{BaldanC14}. As future work, we would like to understand whether they can be captured in our coalgebraic setting. Another open research question is how to obtain coalgebraic models for other notions of causal bisimulation, such as hereditary history preserving bisimulation.

\bibliographystyle{plain}
\bibliography{biblio}

\appendix

\section{Proofs}

We first introduce some technical lemmata. Then we give proofs for the claims in the paper.

\subsection{Additional lemmata}

\begin{lemma}
\label{lem:clts-clos}
Let $O_1,O_2$ be finite $Act$-labeled posets and let $\sigma \colon O_1 \to O_2$ be an order-embedding. Then:
\begin{enumerate}[$(i)$]
    \item 
    $O_1 \rhd c \ctrans{K}{e}{a} \delta(O_1,K,e_a) \rhd c'$ implies $O_2 \rhd \dclos{(c\sigma)}{O_2} \ctrans{\sigma(K)}{e'_a}{} O_2' \rhd \dclos{(c'\funext{\sigma}{e'_a}{e_a})}{O_2'}$, for any $e' \notin X_{O_2}$, with $O_2' = \delta(O_2,\sigma(K),e'_a)$;
    \item 
    $O_2 \rhd c \ctrans{K}{e}{a} \delta(O_2,K,e_a) \rhd c'$ implies $O_1 \rhd c'' \ctrans{K'}{e'_a}{} \delta(O_1,K',e'_a) \rhd c'''$, with $c''\sigma = c$, $\sigma(K') = K$ and $c''' \funext{\sigma}{e_a}{e'_a} = c'$, for any $e' \notin X_{O_1}$.
\end{enumerate}
\end{lemma}

\begin{proof}
We prove item (i), the other one is analogous. Suppose $O_1 \rhd c \ctrans{K}{e}{a} \delta(O_1,K,e_a) \rhd c'$ is derived from the rule of \autoref{def:clts} as follows
    \[
        \frac{
            t \in T
            \quad
            |c_1| = \pre{t}
            \quad
            a = l(t)
            \quad
            e \notin X_{O_1}
            \quad
            K = \max_{O_1} \causes(c_1)
        }
        {
            O_1 \rhd c_1 \cup c_2 \ctrans{K}{e}{a} \delta(O_1,K,e_a) \rhd (\causes(c_1) \cup \{e_a\} \caus \post{t} ) \cup c_2
        }
    \]
where $c=c_1 \cup c_2$ and $c' = (\causes(c_1) \cup \{e_a\} \caus \post{t} ) \cup c_2$. Clearly we have $\dclos{(c\sigma)}{O_2} = \dclos{(c_1\sigma)}{O_2} \cup \dclos{(c_2\sigma)}{O_2}$, with $|\dclos{(c_1\sigma)}{O_2}| = |c_1|$, because $\sigma$ only affects events, not tokens. Moreover, it can be easily verified that $\max_{O_2} \causes(\dclos{(c_1\sigma)}{O_2}) = \sigma(\max_{O_1} \causes(c_1)) = \sigma(K)$. In fact, causes of $\dclos{(c_1\sigma)}{O_2}$ are: those of $c_1\sigma$, related exactly as their counterimages, due to $\sigma$ preserving and reflecting order; additional causes, smaller than those of $c_1\sigma$, added by the closure. Therefore we can again apply the rule as follows
\[
        \frac{
            t \in T
            \quad
            |\dclos{(c_1\sigma)}{O_2}| = \pre{t}
            \quad
            a = l(t)
            \quad
            e' \notin X_{O_2}
            \quad
            \sigma(K) = \max_{O_2} \causes(\dclos{(c_1\sigma)}{O_2})        }
        {
            O_2 \rhd \dclos{(c_1\sigma)}{O_2} \cup \dclos{(c_2\sigma)}{O_2} \ctrans{\sigma(K)}{e'_a}{} O_2' \rhd (\causes(\dclos{(c_1\sigma)}{O_2}) \cup \{e'_a\} \caus \post{t} ) \cup \dclos{(c_2\sigma)}{O_2}
        }
\]
where $O_2' = \delta(O_2,\sigma(K),e'_a)$. Now, observe that, by definition of $\delta$, we have 
\[
    \causes(\dclos{(c_1\sigma)}{O_2}) \subseteq \causes(\dclos{(c_1\sigma)}{O'_2})
    \qquad
    \dclos{\{e'_a\}}{O'_2} =\causes(\dclos{(c_1\sigma)}{O'_2}) \cup \{e'_a\}    
\]
which implies
\begin{align*}
\causes(\dclos{(c_1\sigma)}{O_2}) \cup \{e'_a\} \caus \post{t}
&= \dclos{(\causes(c_1 \sigma ) \cup \{e'_a\})}{O'_2} \caus \post{t}\\
&= \dclos{(\causes(c_1) \cup \{e_a\}) \funext{\sigma}{e'_a}{e_a}}{O_2'} \caus \post{t}
\\
&= \dclos{(\causes(c_1) \cup \{e_a\} \caus \post{t}) \funext{\sigma}{e'_a}{e_a}}{O_2'}
\end{align*}
From this equation, and from $\dclos{(c_2\sigma)}{O_2} = \dclos{(c_2\funext{\sigma}{e'_a}{e_a})}{O_2'}$, because $e_a \notin \causes(c_2)$, it follows that the continuation derived from the above rule has the required shape.
\end{proof}

\begin{lemma}
Let $\sigma \colon O \to O'$ be an isomorphism. Then $O \rhd c_1 \cbisim O \rhd c_1$ implies $O' \rhd c_1\sigma \cbisim O' \rhd c_2\sigma$.
\label{lem:cbisim-iso} 
\end{lemma}
\begin{proof}
We will prove that the following relation is a \cb-bisimulation
\[
    R_{O'}= \{ (O' \rhd c_1\sigma,O' \rhd c_2\sigma) \mid O \rhd c_1 \cbisim O \rhd c_2 , \text{$\sigma \colon O \to O'$ is an isomorphism} \}
\]
Take $(O' \rhd c_1\sigma,O' \rhd c_2\sigma) \in R_{O'}$ and 
\[
    O' \rhd c_1\sigma \ctrans{K'}{e'}{a} \delta(O',K',e'_a) \rhd c_1'
\]
We have to find a simulating transition of $O' \rhd c_2\sigma$. Let $e \notin X_O$. We can apply \autoref{lem:ctrans-iso}, using the isomorphism $\funext{\sigma^{-1}}{e_a}{e'_a}$, and get
\[
    O \rhd c_1 \ctrans{\sigma^{-1}(K)}{e}{a} \delta(O,\sigma^{-1}(K),e_a) \rhd c_1'\funext{\sigma^{-1}}{e_a}{e'_a}
\]
Since $O \rhd c_1 \cbisim O \rhd c_2$, there is a simulating transition
\[
    O \rhd c_2 \ctrans{\sigma^{-1}(K)}{e}{a} \delta(O,\sigma^{-1}(K),e_a) \rhd c_2' \enspace .
\]
Applying again \autoref{lem:ctrans-iso} with $\funext{\sigma}{e'_a}{e_a}$ to this transition, we get 
\[
    O' \rhd c_2\sigma \ctrans{K'}{e'}{a} \delta(O',K',e'_a) \rhd c_2'\funext{\sigma}{e'_a}{e_a} \enspace .
\] 
This is the required simulating transition. In fact, since
\[
    \delta(O,\sigma^{-1}(K),e_a) \rhd c_1'\funext{\sigma^{-1}}{e_a}{e'_a} \;\;\cbisim \;\; \delta(O,\sigma^{-1}(K),e_a) \rhd c_2'
\] 
and $\funext{\sigma}{e'_a}{e_a}$ is an isomorphism, by definition of $R_{O'}$ we have
\[
    ( \; \delta(O',K',e'_a) \rhd c_1' \;\;,\;\; \delta(O',K',e'_a) \rhd c_2'\funext{\sigma}{e'_a}{e_a} \;) \in R_{O'} \enspace .
\]
\end{proof}

\begin{lemma}
Let $O \rhd c_1$ and $O \rhd c_2$ be abstract P-markings. Then $O \rhd c_1 \cbisim O \rhd c_2$ if and only if $O \rhd c_1 \acbisim O \rhd c_2$.
\label{lem:apm-cbisim-acbisim}
\end{lemma}
\begin{proof}
We show the left-to-right implication, the other one is analogous. We prove that the following relation is an \acb-bisimulation
    \[
        R_O = \{ ( O \rhd c_1,O \rhd c_2) \mid O \rhd c_1 \cbisim O \rhd c_2 \}
    \]
    Take $( O \rhd c_1,O \rhd c_2) \in R_O$ and suppose
    \[
        O \rhd c_1 \actrans{K}{a} \delta(O,K,a) \rhd c_1'
    \]
    then we must find a simulating transition of $O \rhd c_2$. By \autoref{def:aclts}, the above transition can be derived from
    \[
        O \rhd c_1 \ctrans{K}{e}{a} \delta(O,K,e_a) \rhd c_1''
    \]
    with $c''_1 \funext{old(O,K,e_a)}{new(O,K,e_a)}{e_a} = c'_1$.
    Since $O \rhd c_1 \cbisim O \rhd c_2$ by hypothesis, this transition can be simulated by 
    \[
        O \rhd c_2 \ctrans{K}{e}{a} \delta(O,K,e_a) \rhd c_2'' \enspace .
    \]
    Applying again \autoref{def:aclts}, we get the required transition
    \[
        O \rhd c_2 \actrans{K}{a} \delta(O,K,a) \rhd c_2'' (\funext{old(O,K,e_a)}{new(O,K,e_a)}{e_a}).
    \]
    In fact, from $\delta(O,K,e_a) \rhd c_1'' \cbisim \delta(O,K,e_a) \rhd c_2''$, using \autoref{lem:cbisim-iso} with the isomorphism $\funext{old(O,K,e_a)}{new(O,K,e_a)}{e_a}$, we get 
    \[
        \delta(O,K,a) \rhd c_1' \;\cbisim\; \delta(O,K,a) \rhd c_2'' (\funext{old(O,K,e_a)}{new(O,K,e_a)}{e_a})
    \]
    and we can conclude that these P-markings are related by $R_{\delta(O,K,a)}$, by its definition.
\end{proof}

\begin{lemma}
Let $\xymatrix@1{O_2 & O \ar[l]_{\sigma_2} \ar[r]^{\sigma_1} & O_1}$ be a span in $\catO$ and let 
\[
    \xymatrix{
        O  \ar[d]_{\sigma_2} \ar[r]^{\sigma_1} & O_1 \ar[d]^{p_1} \\
        O_2 \ar[r]_{p_2} & O_3 \pushoutcorner
    }
\]
be its pushout in $\catPO$. Then it is also a pushout in $\cat{O}$, with 
\[
    l_{O_3}(x) = 
    \begin{cases}
        l_{O_1}(y) & x  = p_1(y) \\
        l_{O_2}(y) & x = p_2(y)
    \end{cases}
\]
\label{lem:o-po-push}
\end{lemma}
\begin{proof}
In \cite[Lemma 8]{ACTA} we have proved that pushouts in $\catPO$ are computed as in $\catGraph$, plus transitive closure of the pushout object. We will use this fact to prove our claim.

First of all, we check that $l_{O_3}$ is well-defined. We only have to verify that its definition is correct for $x = p_1(y_1) = p_2(y_2)$. If $p_1(y_1) = p_2(y_2)$ then $y_1$ and $y_2$ are images via $\sigma_1$ and $\sigma_2$ of the same element of $O$, by definition of pushout in $\catGraph$. Since $\sigma_1$ and $\sigma_2$ preserve labels, we must have $l_{O_1}(y_1) = l_{O_2}(y_2)$, so $l_{O_3}(x)$ is well-defined on $x$. 

Preservation of labels by $p_1$ and $p_2$ follows immediately from the definition of $l_{O_3}$.

Now we prove that the square is indeed a pushout in $\cat{O}$.
Consider the following situation:
\[
    \xymatrix{
        O \ar[d]_{\sigma_2} \ar[r]^{\sigma_1} & O_1 \ar[d]^{p_1} \ar@/^1pc/[ddr]^{q_1}\\
        O_2 \ar@/_1pc/[drr]_{q_2} \ar[r]_{p_2} & O_3 \pushoutcorner \ar@{..>}[dr]^{m}\\
        && O_4
    }
\]
We have to check that, when $q_1$ and $q_2$ preserve labels, also the unique mediating morphism $m$, as computed in $\catPO$, does. We prove it by contradiction. Suppose $m$ does not preserve labels, then there exists $x \in X_{O_3}$ such that $l_{O_4}(m(x)) \neq l_{O_3}(x)$. Suppose $x$ is image of $y \in X_{O_1}$ via $p_1$ (the case $y \in X_{O_2}$ and $x = p_2(y)$ is analogous). Then we have
\begin{align*}
    l_{O_1}(y) &= l_{O_3}(x) && (\text{by $p_1$ preserving labels})
    \\
    &\neq l_{O_4}(m(x)) && (\text{by hypothesis})
    \\
    &= l_{O_4}(q_1(y)) && (\text{by $q_1 = m \circ p_1$})
\end{align*}
which implies that $q_1$ does not preserve labels, a contradiction.
\end{proof}

\subsection{Main proofs}

\begin{proof}[Proof of \autoref{lem:ctrans-iso}]
It is just a corollary of \autoref{lem:clts-clos}.
\end{proof}

\begin{proof}[Proof of \autoref{prop:aclts-clos}]
We prove (i), the other point is similar. Suppose 
\[
    O \rhd c \actrans{K}{a} \delta(O,K,a) \rhd c'.
\] 
Then, by \autoref{def:aclts}, this transition can be derived from
\[
    O \rhd c \ctrans{K}{e}{a} \delta(O,K,e_a) \rhd c''
\]
with $c' = c'' \funext{old(O,K,e_a)}{new(O,K,e_a)}{e_a}$, for any $e \notin X_O$. Suppose $e \notin X_{O'}$. By \autoref{lem:clts-clos}(i), we have
\[
    O' \rhd c\sigma \ctrans{\sigma(K)}{e}{a} \delta(O',\sigma(K),e_a) \rhd \dclos{(c''\funext{\sigma}{e_a}{e_a})}{\delta(O',\sigma(K),e_a)}
\]
from which, using \autoref{def:aclts}, we get
\[
    O' \rhd c\sigma \actrans{\sigma(K)}{a} \delta(O',\sigma(K),a) \rhd \dclos{(c''\funext{\sigma}{e_a}{e_a})}{\delta(O',\sigma(K),e_a)} \omega
\]
where $\omega =  \funext{old(O',\sigma(K),e_a)}{new(O',\sigma(K),e_a)}{e_a}$. We have to prove that the continuation of this transition has the required form.

It is immediate to verify that, for any isomorphism $\sigma \colon O \to O'$ and causal marking $c$ such that $\causes(c) \subseteq |O|$, we have 
\[
    \dclos{(c\sigma)}{O'} = \dclos{c}{O}\sigma
\]
which, for $\sigma = \omega$, implies
\begin{equation}
\dclos{(c''\funext{\sigma}{e_a}{e_a})}{\delta(O',\sigma(K),e_a)}\omega
   =
   \dclos{(c''\funext{\sigma}{e_a}{e_a}\omega)}{\delta(O',\sigma(K),e_a)} \enspace .
    \label{eq:cont}
\end{equation}
Now, observe that, by the definition of $\ext{\sigma}$ we have
\[
    \funext{\sigma}{e_a}{e_a}\omega = \funext{old(O,K,e_a)}{new(O,K,e_a)}{e_a}\ext{\sigma}
\]
therefore \eqref{eq:cont} is equal to 
\[
    \dclos{(c''\funext{old(O,K,e_a)}{new(O,K,e_a)}{e_a}\ext{\sigma})}{\delta(O',\sigma(K),a)}
    =
    \dclos{(c'\ext{\sigma})}{\delta(O',\sigma(K),a)}
\]
as required.
\end{proof}

\begin{proof}[Proof of \autoref{thm:acbisim-cbisim}]
Both implications can be proved by combining \autoref{lem:apm-cbisim-acbisim} and \autoref{lem:cbisim-iso}.
\end{proof}

\begin{proof}[Proof of \autoref{thm:corr-ic-ac}]
This is proved as \cite[Theorem 2]{ACTA}, where specific choices for abstract posets and $old$ and $new$ maps are made in order to accommodate Darondeau-Degano LTSs. The proof is exactly the same, where each specific operation is replaced by its general version described in this paper.
\end{proof}

\begin{proof}[Proof of \autoref{prop:reach-size}]
Take $c \in \fires{n_0}$. Then its tokens have been created by at most $|c|$ transitions. Since we only take immediate causes, i.e., events generated when those transitions were fired, each $O \RHD c$ is such that $|O|$ contains at most $|c|$ events. $O$ can be any poset on those events but, since posets of minimal P-markings must be abstract, there are finitely-many such posets. 
\end{proof}

\begin{proof}[Proof of \autoref{lem:rclts-prop}] 
\hfill
\begin{enumerate}[$(i)$]
    \item Immediately from the fact that any path from $\emptyset \rhd \emptyset \caus m_0$ to $O_c \rhd c$ builds $O_c$ and $c$ incrementally, adding one event for each transition.
    \item Suppose there are two parallel transitions from $O \rhd c$ to $O' \rhd c'$, with labels $a$ and $b$. Then $O' = \delta(O,K,e_a) = \delta(O,K',e'_b)$, which can only happen when $K = K'$ and $e_a = e'_b$, i.e., when the two transitions coincide.
    
Suppose there is a directed cycle starting and ending at $O \rhd c$. Each transition in the cycle would add a new event to $O$, so the final state would be $O' \rhd c$, with $O'$ a strict superposet of $O$, a contradiction.
\end{enumerate}
\end{proof}

\begin{proof}[Proof of \autoref{thm:bs-c-bisim}]
\hfill
\begin{enumerate}[$(i)$]
    \item Consider a transition $c_1 \bstrans{\bsc}{a} c_1'$ and suppose the corresponding transition in \rclts{} is
    \[
        O_{c_1} \rhd c_1 \ctrans{K}{e}{a} \delta(O_{c_1},K,e_a) \rhd c_1'
    \] 
    Now, observe that there is a trivial embedding of $O_{c_1}$ into $O$. In fact, causes of $c_1$ are down-closed w.r.t.\ both posets, so $O_{c_1}$ must be a prefix of $O$. Then, using \autoref{lem:clts-clos}(i) and the embedding $O_{c_1} \subto O$ on the above transition, we get
    \[
        O \rhd c_1 \ctrans{K}{e'_a}{} \delta(O,K,e'_a) \rhd c'_1\funext{}{e'_a}{e_a}
    \]
    for any $e' \notin X_O$. By the hypothesis $O \rhd c_1 \cbisim O \rhd c_2$, this transition can be simulated by
    \[
        O \rhd c_2 \ctrans{K}{e'_a}{} \delta(O,K,e'_a) \rhd c'_2
    \]
    with $\delta(O,K,e'_a) \rhd c'_1\funext{}{e'_a}{e_a} \cbisim \delta(O,K,e'_a) \rhd c'_2$. Using \autoref{lem:clts-clos}(ii) on the embedding of $O_{c_2}$ into $O$, and noting that $e' \notin X_{O_{c_2}}$, we recover a transition
    \[
        O_{c_2} \rhd c_2 \ctrans{K}{e'}{a} \delta(O_{c_2},K,e'_a) \rhd c'_2
    \]
    and from this, using the rule in \autoref{prop:bsc}, we get $c_2 \bstrans{\bsc}{a} c_2'$. In order to show that this transition simulates $c_1 \bstrans{\bsc}{a} c_1'$, we have to find an isomorphism $\sigma' \colon O_{c_1'} \to O_{c_2'}$ such that the following diagram commutes
    \[
        \xymatrix@C+3ex{
            O_{c_1} \ar@{^{(}->}[r]^{\phic_{c_1,c_1'}} \ar[d]_{\sigma} & O_{c_1'} \ar[d]^{\sigma'}\\
            O_{c_2} \ar@{^{(}->}[r]_{\phic_{c_2,c_2'}} &
            O_{c_2'}
        }
    \]
    We can define $\sigma'(x)$ as $\sigma(x)$ if $x \in |O_{c_1}|$ and as $e'_a$ if $x = e_a$.
    \item We want to prove that the following relation is an \acb-bisimulation
    \[
        R_{O_{c_2}} = \{ (O_{c_2} \rhd c_1\sigma,O \rhd c_2) \mid c_1 \bsbisim^\sigma c_2 \}
    \]
    Suppose $c_1 \bsbisim^\sigma c_2$ and
    \begin{equation}
        O_{c_2} \rhd c_1\sigma \ctrans{K}{e}{a} \delta(O_{c_2},K,e_a) \rhd c_1' .
        \label{c1-ctrans}
    \end{equation}
    We have to find a simulating transition of $O_{c_2} \rhd c_2$. Applying \autoref{lem:ctrans-iso} to the last transition, with isomorphism $\sigma^{-1}$, we get
    \[
        O_{c_1} \rhd c_1 \ctrans{\sigma^{-1}(K)}{e'_a}{} \delta(O_{c_1},\sigma^{-1}(K),e'_a) \rhd c_1''
    \]
    where $c_1'' = c_1' \funext{\sigma^{-1}}{e'_a}{e_a}$, for any $e' \notin X_{O_{c_1}}$. This transition corresponds, via \autoref{prop:bsc}, to the following transition in $\bsc$ 
    \[
        c_1 \bstrans{\bsc}{a} c_1''
    \]
    which, by the hypothesis $c_1 \bsbisim^\sigma c_2$, can be simulated by 
    \begin{equation}
        c_2 \bstrans{\bsc}{a} c_2'
        \label{c2-bstrans}
    \end{equation}
    with $c_1' \bsbisim^{\sigma'} c_1''$ such that 
    \begin{equation}
        \phic_{c_2,c_2'} \circ \sigma = \sigma'\ \circ \phic_{c_1,c_1''}
        \label{eq:scomm}
    \end{equation}
     Now, suppose for simplicity $\{ e_a \} =  |O_{c'_2}| \setminus |O_{c_2}|$ (the general case where $|O_{c'_2}| \setminus |O_{c_2}|$ contains any event fresh w.r.t.\ $O_{c_2}$ requires minor changes). By definition of $\phic_{c_2,c_2'}$ and $\phic_{c_1,c_1''}$, and by \eqref{eq:scomm}, $\sigma'$ should act as $\sigma$ on $O_{c_1}$, so $\sigma' = \funext{\sigma}{e_a}{e'_a}$. Moreover, since $\sigma'$ is an isomorphism, we have that the maximal causes of $e'_a$, namely $\sigma^{-1}(K)$, are mapped by $\sigma'$ to the maximal causes of $e''_a$, which then are $\sigma'(\sigma^{-1}(K)) = \sigma(\sigma^{-1}(K)) = K$, where the first equation follows from $e'_a \notin \sigma^{-1}(K)$. Therefore $O_{c_2'} = \delta(O_{c_2},K,e_a)$ and \eqref{c2-bstrans} is derived, using \autoref{prop:bsc}, from
    \[
        O_{c_2} \rhd c_2 \ctrans{K}{e_a}{} \delta(O_{c_2},K,e_a) \rhd c_2'
    \]
This transition is the required one simulating \eqref{c1-ctrans}. In fact, $c_1'' \bsbisim^{\sigma'} c_2'$ implies \[
    (\delta(O_{c_2},K,e_a) \rhd c_1''\sigma' , \delta(O_{c_2},K,e_a) \rhd c_2') \in R_{\delta(O_{c_2},K,e_a)}
\]
by definition of $R$, and for the first P-marking we have $c''_1\sigma' = c_1'' \funext{\sigma}{e_a}{e'_a} = ( c_1' \funext{\sigma^{-1}}{e'_a}{e_a})\funext{\sigma}{e'_a}{e_a} = c_1'$, which is the causal marking in the continuation of \eqref{c1-ctrans}.
\end{enumerate}

\end{proof}

\begin{proof}[Proof of \autoref{prop:o-small-pull}]
Smallness follows from skeletality. In \cite{ACTA} we have proved that pullbacks in $\catPOm$ are computed as the category $\catGraph$ of graphs and their homomorphisms. It can be easily verified that, given a cospan $\xymatrix@C=2ex{O_1 \ar[r]^f & O_3 & \ar[l]_{g} O_2}$ in $\catO$, we can forget labels and compute the pullback as in $\catGraph$. In fact, the pullback poset $O$ has an element $y$ for each pair of elements $x_1 \in X_{O_1}$ and $x_2 \in X_{O_2}$ such that $f(x_1) = g(x_2)$. But then, since $f$ and $g$ preserve labels, we must have $l_{O_1}(x_1) = l_{O_2}(x_2) = a$, so $l_O(y) = a$ and the pullback maps preserve labels. It is easy to check that pullback mediating morphisms preserve labels, as they must commute with morphisms with such property.
\end{proof}

\begin{proof}[Proof of \autoref{lem:o-delta-diag}]
In (\cite[Lemma 8]{ACTA}) we have proved that pushouts of order-embeddings in $\catPO$ are commuting squares in $\catPOm$. Therefore we can compute the two pushouts of \eqref{diag:delta} in $\catPO$, take the corresponding commuting squares in $\catPOm$ and then use \autoref{lem:o-po-push} to get labeling functions for their bottom-right corners. Diagrams in $\catPOm$ made of label preserving functions are also diagrams in $\catO$.

Finally, the fact that $\ext{\sigma}$ reflects orders follows from its definition.
\end{proof}

\begin{proof}[Proof of \autoref{prop:B-acc}]
$B$ is obtained by composition and product of accessible functors: $\FinParts$ is known to be accessible; $\lpsh$ is accessible, because it can be regarded as a constant endofunctor on $\pshO$; $\Delta$ is accessible, because it has a right adjoint, namely the functor computing right Kan extensions along $\delta$. 

In order to show that $B$ covers pullbacks, we will show that it has the form $\FinParts \circ B'$, with $B'$ a pullback preserving endofunctor on $\pshO$. The thesis will follow from $\FinParts$ covering pullbacks (see \cite{Staton11}). $\Delta$ has a left adjoint, namely the functor computing left Kan extensions along $\delta$, then it preserves pullbacks; $\lpsh$ can be seen as a constant, hence pullback-preserving, endofunctor on $\pshO$. $B'$ is the product of these two functors, so it preserves pullbacks.
\end{proof}

\begin{proof}[Proof of \autoref{prop:bbsim-obisim-eq}]
Requirement \autoref{def:oibisim}\ref{oibisim-clos} corresponds to the fact that a $B$-bisimulation $R$ on $(P,\rho)$ is a functor and its projections are natural transformations, so we have $(p,q)[\sigma]_R = (p[\sigma]_P,q[\sigma]_P)$, for any morphism $\sigma$ in $\catO$. Requirement \ref{oibisim-sim} corresponds to the fact that $RO$ is ``almost'' an ordinary bisimulation, because computing $\overline{B}R(O)$ essentially amounts to computing $\overline{B_{flts}}(RO)$ (see \autoref{ssec:coalg}) for each $O \in |\catO|$, as images in $\pshO$ are computed pointwise in $\catSet$, with the difference that continuations are not in $RO$, but in $R(\delta O)$.\qed
\end{proof}

\begin{proof}[Proof of \autoref{lem:mfun-pull}]
We have to prove that if the square on the left is a pullback then so is the outer square on the right.
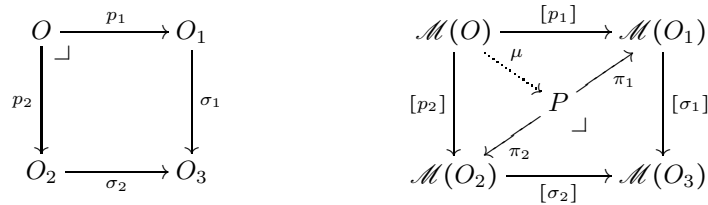
\[
    \xymatrix@=9ex{
        O \pullbackcorner \ar[r]^{p_1} \ar[d]_{p_2} & O_1 \ar[d]^{\sigma_1} \\
        O_2 \ar[r]_{\sigma_2} & O_3
    }
    \hspace{15ex}
    \xymatrix@=3ex{
            \pshMark(O) \ar@{..>}[dr]^{\mu} \ar[rr]^{[p_1]} \ar[dd]_{[p_2]} && \pshMark(O_1) \ar[dd]^{[\sigma_1]} \\
            & P \pullbackcorner[dr] \ar[dl]^{\pi_2} \ar[ur]_{\pi_1} \\
        \pshMark(O_2) \ar[rr]_{[\sigma_2]} && \pshMark(O_3)
    }
\]
In the right diagram, let $P$ be the pullback in $\catSet$ of $[\sigma_1]$ and $[\sigma_2]$, namely 
\[
    P = \{ (c_1,c_2) \mid c_1[\sigma_1] = c_2[\sigma_2] \}
\]
We will show that that the mediating morphism $\mu$ is an isomorphism, which implies that $\pshMark(O)$ is a pullback object.

Take $(c_1,c_2) \in P$ and $c = c_1[\sigma_1] = c_2[\sigma_2]$. Then these causal markings must be of the form
\[
    c_1 = \{ K_1 \caus s_1, \dots, K_n \caus s_n \} 
    \qquad
    c_2 = \{ H_1 \caus s_1, \dots , H_n \caus s_n \}
    \qquad
    c = \{ L_1 \caus s_1 , \dots , L_n \caus s_n \}
\]
because $[\sigma_1]$ and $[\sigma_2]$ do not affect tokens. Moreover, we must have
\[
    L_i = \dclos{\sigma_1(K_i)}{O_1} = \dclos{\sigma_2(H_i)}{O_2} \qquad (i=1,\dots,n)
\]
by definition of the action of $\pshMark$ on morphisms, and in particular
\[
    \max_{O_1} \sigma_1(K_i) = \max_{O_2} \sigma_2( H_i ) = max_{O_3} L_i
\]
because $K_i,H_i$ and $L_i$ are down-closed sets, so they coincide with the closure of their maxima. It is easy to check that order-preserving and reflecting morphisms preserve maxima, so we have 
\[
    \sigma_1( max_{O_1} K_i) = \max_{O_3} \sigma_1(K_i) = \max_{O_2} \sigma_2( H_i ) = \sigma_2( max_{O_2} H_i) .
\]
Therefore, by definition of pullback in $\catO$ (computed as in $\catGraph$), there are $J_i \subseteq |O|$ such that 
\begin{equation}
    p_1(J_i) = \max_{O_1} K_i \qquad p_2(J_i) = \max_{O_2} H_i
    \label{rel-max}
\end{equation}    
and we can define the following causal marking in $\pshMark(O)$
\[
    c' = \{ \hat{J}_1 \caus s_1 , \dots , \hat{J}_n \caus s_n \}
\]
where $\hat{J}_i = \dclos{J_i}{O}$.

Now, observe that $c'[p_1] = c_1$ and $c'[p_2] = c_2$, because \eqref{rel-max} implies $\dclos{p_1(\hat{J}_i)}{O_1} = K_i$ and $\dclos{p_2(\hat{J}_i)}{O_2} = H_i$. Therefore letting $\mu(c') = (c_1,c_2)$ makes the whole right diagram commute. So far we have proved that $\mu$ is surjective. For injectivity, suppose there is another $c'' \in \pshMark(O)$ such that  $\mu(c'') = (c_1,c_2)$. Since $c''[p_1] = c_1$ and $c''[p_2] = c_2$, $c''$ is again of the form $\{ M_1 \caus s_1 , \dots , M_n \caus s_n \}$, with $\dclos{p_1(M_i)}{O_1} = K_i$. Since also $K_i = \dclos{p_1(\hat{J}_i)}{O_1}$, $M_i$ and $\hat{J}_i$ must have the same set $X$ of maxima. But then we have $M_i = \dclos{X}{O} = \hat{J}_i$, so $c'' = c'$.

\end{proof}

\begin{proof}[Proof of \autoref{thm:bisim-equiv}]
The first item is just an instance of \autoref{prop:bbsim-obisim-eq}.

For the second item, we shall show that $R$ is an \acb-bisimulation closed under order-embeddings if and only if it is a $\catO$-indexed bisimulation:
\begin{itemize}
	\item[$\implies$:] take $(O \rhd c,O \rhd \tilde{c}) \in R_O$ and suppose
	\begin{equation}
		O \rhd c \acTrans{K}{a} \deltaFun(O) \rhd c'.
		\label{tr:ac}
	\end{equation}
	Then, by \autoref{def:dd-ilts}, there is 
	\[
		O \rhd c \actrans{K}{a} \delta(O,K,a) \rhd c''
	\]
	such that $c' = c''[\epsilon(O,K,a)]$. Since $R$ is an \acb-bisimulation, there is
	\[
		O \rhd \tilde{c} \actrans{K}{a} \delta(O,K,a) \rhd \tilde{c}'
	\]
	such that $(\delta(O,K,a) \rhd c'',\delta(O,K,a) \rhd \tilde{c}') \in R_{\delta(O,K,a)}$. Again by \autoref{def:dd-ilts}, from the last transition we get
	\[	
		O \rhd \tilde{c} \acTrans{K}{a} \deltaFun(O) \rhd \tilde{c}'[\epsilon(O,K,a)].
	\]
	This is a simulating transition for (\ref{tr:ac}), because
	$(\delta(O,K,a) \rhd c'',\delta(O,K,a) \rhd \tilde{c}') \in R_{\delta(O,K,a)}$ implies $(\deltaFun(O) \rhd c',\deltaFun(O) \rhd \tilde{c}'[\epsilon(O,K,a)]) \in R_{\deltaFun(O)}$, by closure of $R$ under order-embeddings.
	\item[$\impliedby$:] analogous to the previous point. Closure under order-embeddings of $R$ follows from \autoref{def:oibisim}\ref{oibisim-clos}.
\end{itemize}
\end{proof}

\begin{proof}[Proof of \autoref{prop:res-pull}]
Analogous to the proof of \cite[Proposition 8]{ACTA}.
\end{proof}

\end{document}